\documentclass[12pt, twoside]{article}
\usepackage{amsmath,amsthm,amssymb}
\usepackage{times}
\usepackage{enumerate}
\usepackage{bm}
\usepackage[all]{xy}
\usepackage{mathrsfs}
\usepackage{amscd}
\usepackage{color}
\def\titlerunning#1{\gdef\titrun{#1}}
\makeatletter
\def\author#1{\gdef\autrun{\def\and{\unskip, }#1}\gdef\@author{#1}}
\def\address#1{{\def\and{\\\hspace*{18pt}}\renewcommand{\thefootnote}{}%
\footnote {#1}}%
\markboth{\autrun}{\titrun}}
\makeatother
\def\email#1{e-mail: #1}
\def\subjclass#1{{\renewcommand{\thefootnote}{}%
\footnote{\emph{Mathematics Subject Classification (2010):} #1}}}
\def\keywords#1{\par\medskip
\noindent\textbf{Keywords.} #1}

\newtheorem{thm}{Theorem}[section]
\newtheorem{cor}[thm]{Corollary}
\newtheorem{lem}[thm]{Lemma}
\newtheorem{prop}[thm]{Proposition}
\theoremstyle{definition}
\newtheorem{defi}[thm]{Definition}
\newtheorem{rem}[thm]{Remark}

\numberwithin{equation}{section}

\xyoption{all}

\def \N {\mathbb{N}}
\def \C {\mathbb{C}}

\def \D {\mathcal{D}}

\def \a {\alpha }
\def \b {\beta}

\def \de {\delta}
\def \De {\Delta}
\def \la {\lambda}
\def \La {\Lambda}
\def\w {\omega}
\def\Om{\Omega}

\def\pa{\partial}
\def\na {\nabla}
\def\Ga{\Gamma}
\def\ra{\rightarrow}
\def\en{\textendash}
\begin{document}
\baselineskip=17pt

\titlerunning{On a topology property for the moduli space of Kapustin$\textendash$Witten equations}
\title{On a topology property for the moduli space of Kapustin$\textendash$Witten equations}

\author{Teng Huang}

\date{}

\maketitle

\address{T. Huang: School of Mathematical Sciences, University of Science and Technology of China, Hefei, 230026, P. R.  China; \email{htmath@ustc.edu.cn; htustc@gmail.com}}
\begin{abstract}
In this article, we study the Kapustin$\textendash$Witten equations on a closed, simply-connected, four-dimensional manifold which were introduced by Kapustin and Witten. We use the Taubes' compactness theorem \cite{T1} to prove that if $(A,\phi)$ is a smooth solution of Kapustin$\en$Witten equations and the connection $A$ is closed to a $generic$ ASD connection $A_{\infty}$, then $(A,\phi)$ must be a trivial solution. We also prove that the moduli space of the solutions of Kapustin$\en$Witten equations is non-connected if the connections on the compactification of moduli space of ASD connections are all $generic$. At last, we extend the results for the Kapustin$\en$Witten equations to other equations on gauge theory such as the Hitchin$\en$Simpson equations and Vafa$\en$Witten on a compact K\"{a}hler surface. 
\end{abstract}
\section{Introduction}
Let $X$ be an closed, oriented, smooth, four-dimensional manifold with a given Riemannian metric $g$, and let $P\rightarrow X$ be a principal $G$-bundle with $G$ being a compact Lie group. We denote by $\mathfrak{g}_{P}$ the adjoint bundle of $P$ and by $\mathcal{A}_{P}$ the space of smooth connections on $P$.  On a $4$-manifold $X$ the Hodge star operator $\ast$ takes $2$-forms to $2$-forms and we have $\ast^{2}=Id_{\Om^{2}}$. The self-dual and anti-self-dual forms, we denoted $\Om^{+}$ and $\Om^{-}$ are defined to be the $\pm$ eigenspace of $\ast$: $\Om^{2}T^{\ast}X=\Om^{+}\oplus\Om^{-}$. 

The Kapustin$\en$Witten equations are defined on a Riemannian $4$-manifold given a principle bundle $P$. For most present considerations, $G$ can be taken to be $SU(2)$ or $SO(3)$. The equations require a pair $(A,\phi)\in\mathcal{A}_{P}\times \Om^{1}(X,\mathfrak{g}_{P})$ satisfies
\begin{equation}\label{82}
\begin{split}
&F^{+}_{A}-(\phi\wedge\phi)^{+}=0,\\
&d_{A}^{\ast}\phi=0,\ (d_{A}\phi)^{-}=0,\\
\end{split}
\end{equation}
where  $F_{A}$ is the curvature two-form of the connection $A$, and superscripts ``-" and ``+" indicate taking the anti-self-dual part or self-dual part respectively. These equations were introduced by Kapustin and Witten \cite{KW}. The motivation is from the viewpoint of $\mathcal{N}=4$ super Yang$\en$Mills theory in four dimensions to study the geometric Langlands program \cite{GW,Hay,KW} and \cite{W1,W2,W3,W4}. One also can see Gagliardo-Uhlenbeck's article\cite{GU}. Note that the equation $d_{A}^{\ast}\phi=0$ makes equations (\ref{82}) an elliptic system after gauge fixing equation. Note also there are no solutions to Kapustin$\en$Witten equations in the case of $p_{1}(P)$ is positive.

In mathematics, the analytic properties of solutions of Kapustin$\en$Witten equations were discussed by Taubes \cite{T1,T2,T3} and Tanaka \cite{TY2}. In \cite{T1}, Taubes studied the Uhlenbeck style compactness problem for $SL(2,\C)$ connections, including solutions to the above equations, on four-manifolds, see \cite{T2,T3}. In \cite{TY2}, Tanaka observed that equations on a compact K\"{a}hler surface are the same as Hitchin$\en$Simpson's equations \cite{Hitchin,Simpson1988} and proved that the singular set introduced by Taubes for the case of Simpson's equations has a structure of a holomorphic subvariety. In \cite{HT}, the author proved that there exist a lower bounded for the $L^{2}$-norm of extra fields unless the connection is anti-self-dual  under some mild conditions on $X,P,g,G$. 

One always uses continuous method to construct the solutions of some gauge-theoretic equations. For example, Freed-Uhlenbeck and Taubes \cite{FU,T4} used this technical to constructed the ASD connections over a closed four-manifold. In this article, we suppose that the pair $(A_{\infty}+a,\phi)$ is a smooth solution of the Kapusitin$\en$Witten equations, where $A_{\infty}$ is an anti-self-dual connection on $P$. Then, the pair $(A_{\infty}+a,\phi)$ satisfies the equations:
\begin{equation}\nonumber
\begin{split}
&d^{+}_{A_{\infty}}a+(a\wedge a)^{+}-(\phi\wedge\phi)^{+}=0,\\ &(d_{A_{\infty}}\phi+[a,\phi])^{-}=0.\\
\end{split}
\end{equation}
We will show that if the connections are in a neighborhood of a $generic$ anti-self-dual connection $A_{\infty}$, then the smooth solutions of Kapustin$\en$Witten equations on a closed four-manifold must be trivial solutions, i.e., $F_{A_{\infty}+a}^{+}=0$ and $\phi=0$.
\begin{thm}\label{T1.1}
Let $X$ be a closed, oriented, simply-connected, four-dimensional manifold with a smooth Riemannain metric $g$, $P\rightarrow X$ be a principal $SU(2)$ or $SO(3)$-bundle with $p_{1}(P)$ negative. Suppose that there exist a $generic$ ASD connection $A_{\infty}$ on $P$, i.e., $H^{0}_{A_{\infty}}=\ker d_{A_{\infty}}|_{\Om^{0}(X,\mathfrak{g}_{P})}=0$ and $H^{2}_{A_{\infty}}=\ker d_{A_{\infty}}^{+}|_{\Om^{2,+}(X,\mathfrak{g}_{P})}=0$, then there is a positive constant $\de=\de(g,P,A_{\infty})$ with following significance. If $(A,\phi)$ is a smooth solution of Kapustin$\en$Witten equations over $X$, then one of following must hold:\\
(1) $F^{+}_{A}=0$ and $\phi=0$;\\
(2) the pair $(A,\phi)$ satisfies
$$dist(A,A_{\infty})\geq\de,$$
where $dist(\cdot,\cdot)$ defined by $$dist(A,B):=\inf_{g\in\mathcal{G}_{P}}\|g^{\ast}(A)-B\|_{L^{2}_{1}(X)},\ \forall\ A,B\in\mathcal{A}_{P}.$$
\end{thm}
\begin{rem}
We describe another result as follows which is a small variation of the result in \cite{HT}. Let the pair $(A,\phi)$ be a smooth solution of Kapustin$\en$Witten equations. If $(A,\phi)$ is closed to  $(A_{\infty},0)$ in $L^{2}_{1}$-topology, i.e, there exist a small enough constant $\de\in(0,1)$ such that $$\inf_{g\in\mathcal{G}_{P}}\|g^{\ast}(A)-A_{\infty}\|_{L^{2}_{1}}+\|\phi\|_{L^{2}_{1}}\leq\de,$$
where $A_{\infty}$ is a $generic$ ASD connection, then $A$ must be ASD and $\phi$ identically zero. We should point out that the result is different to Theorem \ref{T1.1}.  We don't make any assumptions on extra fields in theorem \ref{T1.1}. The following case may happen:  if the connection $A$ tends to a $generic$ ASD connection $A_{\infty}$ in $L^{2}_{1}$-topology (under moduli gauge transformations), then the $L^{2}$-norm of extra field $\|\phi\|_{L^{2}(X)}$ may tends to infinite. But thanks to a compactness theorem of Taubes, we observed that if $\{(A_{i},\phi_{i})\}_{i\in\N}$ is a sequence of smooth solutions of Kapustin$\en$Witten equations, then there exist a subsequence $\Xi\subset\mathbb{N}$ such that the sequence $\{\|\phi_{i}\|_{L^{2}(X)}\}_{i\in\Xi}$ is a bounded subsequence under some conditions, see Lemma \ref{L7}. Using this useful observation, we could prove Theorem \ref{T1.1}.
\end{rem}
We define the gauge-equivariant map
$$KW:\mathcal{A}_{P}\times\Om^{1}(X,\mathfrak{g}_{P})\rightarrow\Om^{2,+}(X,\mathfrak{g}_{P})\times\Om^{2,-}(X,\mathfrak{g}_{P}),$$
$$KW(A,\phi)=\big{(}F_{A}^{+}-(\phi\wedge\phi)^{+},(d_{A}\phi)^{-}\big{)}.$$
We denote by
$$M_{KW}(P,g):=\{(A,\phi)\in\mathcal{A}_{P}\times\Om^{1}(X,\mathfrak{g}_{P}\mid KW(A,\phi)=0\}/\mathcal{G}_{P}$$
the moduli space of solutions of Kapustin$\en$Witten equations. The moduli space $M_{ASD}$ of all ASD connections can be embedded into $M_{KW}$ via $A_{\infty}\mapsto(A_{\infty},0)$,\ $A_{\infty}$ is an ASD connection on $P$. We also denote $\bar{M}_{ASD}$ by the compactification of moduli space of ASD connection.

Following the idea of Donaldson on \cite[Section 4.2.1]{DK}, we write $[A,\phi]$ for the equivalence class of a pair $(A,\phi)$, which is a point in the moduli space $M_{KW}$. We set,
$$\|(A_{1},\phi_{1})-(A_{2},\phi_{2})\|^{2}=\|A_{1}-A_{2}\|^{2}_{L^{2}_{1}(X)}+\|\phi_{1}-\phi_{2}\|^{2}_{L^{2}_{1}(X)},$$
is preserved by the action of $\mathcal{G}_{P}$, so descends to define a distance function on $M_{KW}$:
$$dist\big{(}[A_{1},\phi_{1}],[A_{2},\phi_{2}]\big{)}:=\inf_{g\in\mathcal{G}_{P}}\|(A_{1},\phi_{1})-g^{\ast}(A_{2},\phi_{2})\|.$$
Let the pair $(A,\phi)$ be a smooth solution of Decoupled Kapustin$\en$Witten equations over a closed, simply-connected, four-manifold. We observe that if the connection $A$ is irreducible, then the extra field $\phi$  vanishes, see Theorem \ref{T2.6}. Hence, we can denote by 
\begin{equation*}
\begin{split}
dist([A,\phi],M_{ASD}):&=\inf_{g\in\mathcal{G}_{P},A_{\infty}\in M_{ASD}}\|g^{\ast}(A,\phi)-(A_{\infty},0)\|\\
&=\inf_{g\in\mathcal{G}_{P},A_{\infty}\in M_{ASD}}\big{(}\|g^{\ast}(A)-A_{\infty}\|^{2}_{L^{2}_{1}(X)}+\|\phi\|^{2}_{L^{2}_{1}(X)}\big{)}^{\frac{1}{2}},\\
\end{split}
\end{equation*}
the distance between $M_{ASD}$ and $M_{KW}\backslash M_{ASD}$.
\begin{thm}\label{T1.2}
Let $X$ be a closed, oriented, simply-connected, four-dimensional manifold with a smooth Riemannain metric $g$, $P\rightarrow X$ be a principal $SU(2)$ or $SO(3)$-bundle with $p_{1}(P)$ negative. Suppose that the connections in  $\bar{M}_{ASD}(P,g)$ are all $generic$, i.e.,  $H^{0}_{A_{\infty}}=\ker d_{A_{\infty}}|_{\Om^{0}(X,\mathfrak{g}_{P})}=0$ and $H^{2}_{A_{\infty}}=\ker d_{A_{\infty}}^{+}|_{\Om^{2,+}(X,\mathfrak{g}_{P})}=0$ for any connection $[A_{\infty}]\in\bar{M}_{ASD}(P,g)$, then there is a positive constant $\de=\de(g,P)$ with following significance. If $(A,\phi)$ is a smooth solution of Kapustin$\en$Witten equations, then one of following must hold:\\
(1) $F^{+}_{A}=0$ and $\phi=0$;\\
(2) the pair $(A,\phi)$ satisfies
$$2\|\phi\|^{2}_{L^{2}}\geq\|F^{+}_{A}\|_{L^{2}(X)}\geq\de.$$
In particular, there is a positive constant $\tilde{\de}=\tilde{\de}(g,P)$ such that
$$dist(A,M_{ASD}):=\inf_{g\in\mathcal{G}_{P},A_{\infty}\in M_{ASD}}\|g^{\ast}(A)-A_{\infty}\|_{L^{2}_{1}(X)}\geq\tilde{\de},$$
unless $A$ is anti-self-dual  with respect to $g$.
\end{thm}
We now take $X$ to be a compact K\"{a}hler surface with K\"{a}hler form $\w$, and $E$ to be a principal $G$-bundle over $X$. Tanaka observed that the Kapustin$\en$Witten equations are the same as Hitchin$\en$Simpson equations.  Furthermore, if $A$ is an $SU(N)$-ASD connection on $E$, then  
$$\ker d^{+}_{A}d^{+,\ast}_{A}|_{\Om^{2,+}(X,adE)}\cong \ker(\bar{\pa}_{A}\bar{\pa}_{A}^{\ast}+\bar{\pa}_{A}^{\ast}\bar{\pa}_{A})|_{\Om^{0,2}(X,adE)}\oplus\ker d_{A}|_{\Om^{0}(X,adE)},$$
see \cite[Proposition 2.3]{Itoh} or \cite[Chapter IV]{Friedman-Morgan} . It's difficult to addition certain mild conditions to ensure  $\ker(\bar{\pa}_{A}\bar{\pa}_{A}^{\ast}+\bar{\pa}_{A}^{\ast}\bar{\pa}_{A})|_{\Om^{0,2}(X,adE)}$ and $\ker d_{A}|_{\Om^{0}(X,adE)}$ vanish at some time. But one can see that $\ker d_{A}|_{\Om^{0}(X,adE)}=0$ is equivalent to the connection $A$ is irreducible.  It is convenient to introduce the $c$-$generic$ metric which ensures the connections on the compactification of moduli space of ASD connections on $E$, $\bar{M}(E,g)$,  are irreducible, one also can see \cite[Chapter IV]{Friedman-Morgan}.
\begin{defi}\label{D1}
Let  $X$ be a compact, connected, K\"{a}hler surface, $c$ be a positive integer. We say that a K\"{a}hler metric $g$ on $X$ is $c$-$generic$ if for every $SU(2)$-bundle $E$ over $X$ with $c_{2}(E)\leq c$, there are no reducible ASD connections on $E$. 
\end{defi}
If we suppose that the K\"{a}hler metric $g$ is $c$-$generic$, we could deform the connection $A\in\mathcal{A}^{1,1}$ which obeys $\|\La_{\w}F_{A}\|_{L^{2}(X)}\leq \varepsilon$ to an other connection $A_{\infty}$ satisfies $\La_{\w}F_{A_{\infty}}=0$, where $\varepsilon$ is a suitable sufficiently small positive constant,  see Theorem \ref{T6}. Even though, the connection $A_{\infty}$ may be not an ASD connection, but the $(0,2)$-part $F_{A_{\infty}}^{0,2}$ of the curvature $F_{A_{\infty}}$ would estimated by $\La_{\w}F_{A}$. We then extend the results of Kapustin$\en$Witten equations to the Hitchin$\en$Simpson equations case.
\begin{thm}\label{T1}
Let $X$ be a compact, K\"{a}hler surface with a smooth K\"{a}hler metric $g$, $E$ be a principal $SU(2)$-bundle with $c_{2}(E)=c$ positive. Suppose that $g$ is a $c$-$generic$ metric in the sense of Definition \ref{D1}, then there is a positive constant $C=C(g,E)$ with following significance. If the Higgs pair $(A,\theta)\in\mathcal{A}_{E}^{1,1}\times\Om^{1,0}(X,adE)$ is a smooth solution of Hitchin-Simpson equations, then one of following must hold:\\
(1) $\La_{\w}F_{A}=0$, or\\
(2) the Higgs field $\theta$ satisfies 
$$\|\theta\|_{L^{2}(X)}\geq C.$$
Furthermore, if $X$ is also simply-connected, then the curvature $F_{A}$ obeys
$$\|\La_{\w}F_{A}\|_{L^{2}(X)}\geq \tilde{C}$$
for a positive constant $\tilde{C}=\tilde{C}(g,P)$ unless $\La_{\w}F_{A}=0$.
\end{thm}
We denote by 
$$M_{HS}:=\{(A,\phi)\in\mathcal{A}_{E}^{1,1}\times\Om^{1,0}(X,adE)\mid \La_{\w}(F_{A}+[\theta,\theta^{\ast}])=0, \bar{\pa}_{A}\theta=\theta\wedge\theta=0\}/\mathcal{G}_{E}$$
the moduli space of solutions of Hitchin-Simpson equations. In \cite{Hitchin}, Hitchin proved that  if the bundle $E$ is a rank-2 bundle of odd degree over a Riemannian surface of genus $g>1$, then the moduli space of stable Higgs bundle is connected and simply connected, see \cite[Theorem 7.6]{Hitchin}. Following Theorem \ref{T1}, it implies that the moduli space of stable Higgs bundle on a simply-connected K\"{a}hler surface with a generic K\"{a}hler metric is not connected. It is different from the Riemannian surface case, see Corollary \ref{C4.12}.
\begin{rem}
If the principal $G$-bundle $P\rightarrow X$ over the closed manifold  $X$ with $p_{1}(P)=0$, the solutions $(A,\phi)$ of the Kapustin$\en$Witten equations are flat $G_{\C}$-connections with moment map condition:
\begin{equation*}
\begin{split}
&F_{A}-\phi\wedge\phi=0,\\
&d_{A}^{\ast}\phi=0,\ d_{A}\phi=0.\\
\end{split}
\end{equation*}
The \cite[Proposition 2.2.3]{DK} shows that the gauge-equivalence classes of flat $G$-connections over a connected manifold, $X$, are in one-to-one correspondence with the conjugacy classes of representations $\rho:\pi_{1}(X)\rightarrow G$. If $X$ is also a simply-connected manifold, i.e., $\pi_{1}(X)$ is trivial, then the representations $\rho$ must be a trivial representation. It is no sense to consider Kapustin$\en$Witten equations on a principal $G$-bundle with $p_{1}(P)=0$ over a simply-connected four-manifolds.
\end{rem}
The organization of this paper is as follows. In section 2, we first recall gauge theory in $4$-dimensional manifolds. We also discuss the elliptic deformation complex and Kuranishi model for the Kapustin$\en$Witten equations. In section 3, we recall a vanishing theorem for the extra fields of the decoupled Kapustin$\en$Witten equations. Using an optimal inequality proved by Donaldson, we observe that the $L^{2}$-norm of extra fields of non-trivial solutions of Kapustin$\en$Witten equations have a positive lower bounded. Thanks to Taubes' compactness theorem \cite{T1}, we observe that if  $(A_{i},\phi_{i})$ is a sequence solutions of Kapustin$\en$Witten equations and the connection $A_{i}$ converges to an irreducible ASD connection $A_{\infty}$ in $L^{2}_{1}$ (under moduli gauge transformations), then the sequence $\{\|\phi_{i}\|_{L^{2}(X)}\}$ has a bounded subsequence. At last, we obtain our main result: if the connection $A$ is closed to a $generic$ ASD connection $A_{\infty}$, then $(A,\phi)$ must be a trivial solution. In section 4, we extend the above results to the global case. We will prove that if the connections on the compactification of moduli space of ASD connections, $\bar{M}_{ASD}$, are $generic$, then $\|F^{+}_{A}\|_{L^{2}(X)}$ must has a uniform positive lower bounded. In particular, the subspace of anti-self-dual solutions is not connected to the space of Kapustin$\en$Witten equations. We also given some four-manifolds $X$ with Riemannian metric $g$ and principle $SO(3)$-bundles $P\rightarrow X$ ensure that the connections on  $\bar{M}_{ASD}$ are  $generic$. At last, we construct a gluing theorem for the connection $A$ which obeys $\|\La_{\w}F_{A}\|_{L^{2}(X)}\leq\varepsilon$, see Theorem \ref{T6}. We then extend the results for the Kapustin$\en$Witten equations to other coupled equations for pairs such as the Hitchin$\en$Simpson equations and Vafa$\en$Witten on compact K\"{a}hler surface with a smooth $c$-$generic$ K\"{a}hler metric $g$.
\section{A neighborhood of an ASD connection}
\subsection{Yang-Mills theory on 4-manifolds}
Let $X$ be an oriented, closed, smooth, Riemannian $4$-manifold, $P\rightarrow X$ be a principal $G$-bundle with $G$ being a compact Lie group. The Hodge start operator gives an endomorphism of $\Om^{2}$ with property $\ast^{2}=Id_{\Om^{2}}$. We denote by $\Om^{2,+}$ and $\Om^{2,-}$ the eigenvalues of $+1$ and $-1$. A $2$-form in $\Om^{2,+}$ (or in $\Om^{2,-}$) is called self-dual (or anti-self-dual). Decomposing the curvature $F_{A}$ of a connection $A$ according to the decomposition $\Om^{2}=\Om^{2,+}\oplus\Om^{2,-}$ of the 2-forms into self-dual and anti-self-dual parts. An ASD connection $A$ on $P$ naturally induces the Yang$\en$Mills complex
$$\Om^{0}(X,\mathfrak{g}_{P})\xrightarrow{d_{A}}\Om^{1}(X,\mathfrak{g}_{P})\xrightarrow{d^{+}_{A}}{\Om^{2,+}(X,\mathfrak{g}_{P})}.$$
The $i$-th cohomology group $H^{i}_{A}$ of this complex if finite dimensional and the index $d=h^{0}-h^{1}+h^{2}$ ($h^{i}=dim H^{i}_{A}$ )is given by
$c(G)\kappa(P)-dim G(1-b_{1}+b^{+})$. $H^{0}_{A}$ is the Lie algebra of the stabilizer $\Ga_{A}$, the group of gauge transformation of $P$ fixing by $A$. We call a connection is $generic$ when $H^{0}_{A}=0$ and $H^{2}_{A}=0$. We denote $M_{gen}$ the subset of
$$M_{ASD}:=\{A\in\mathcal{A}_{P}:F_{A}+\ast F_{A}=0\}/\mathcal{G}_{P}$$
of $generic$ ASD connections on $P$. $M_{gen}$ becomes a smooth manifold whose tangent space is $H^{1}_{A}$. $M_{gen}$ consists exactly of all singular points and we have two types according to either case (1) in which $A$ is irreducible $H_{A}^{0}=0$ but $H^{2}_{A}\neq0$ or case (2) in which $A$ is reducible $H_{A}^{0}\neq0$. So if the anti-self-dual connections $[A]\in M_{ASD}$ are all $generic$, the moduli space $M_{ASD}$ is a smooth manifold. Furthermore, We call an ASD connection is $regular$ when $H^{2}_{A}=0$.

Now we will present the Kuranishi complex associated to the Kapustin$\en$Witten equations, see \cite[section 3.1]{HT}. Tanaka obtained an analogous Kuranishi complex to the Vafa$\en$Witten equations \cite{BM,Tanaka2015}. Let $(A,\phi)\in\mathcal{A}_{P}\times\Om^{1}(X,\mathfrak{g}_{P})$ be a smooth solution of the Kapustin$\en$Witten equations. Then, the infinitesimal deformation at $(A,\phi)$ is given by the following:
$$0\ra\Om^{1}(X,\mathfrak{g}_{P})\xrightarrow{d^{0}_{(A,\phi)}}\Om^{1}(X,\mathfrak{g}_{P})\oplus\Om^{1}(X,\mathfrak{g}_{P})
\xrightarrow{d^{1}_{(A,\phi)}}\Om^{2,-}(X,\mathfrak{g}_{P})\oplus\Om^{2,+}(X,\mathfrak{g}_{P})\ra 0,$$
where the map $d^{0}_{(A,\phi)}$ and $d^{1}_{(A,\phi)}$ are defined for $\xi\in\Om^{0}(X,\mathfrak{g}_{P})$ and $(a,b)\in\Om^{1}(X,\mathfrak{g}_{P})\oplus\Om^{1}(X,\mathfrak{g}_{P})$ by 
$$d^{0}_{(A,\phi)}(\xi)=(-d_{A}\xi,[\xi,\phi]),$$
$$d^{1}_{(A,\phi)}(a,b)=\big{(}(d_{A}b+[a,\phi])^{-},(d_{A}a+[b,\phi])^{+}\big{)}.$$
We compute
$$d^{1}_{(A,\phi)}\circ d^{0}_{(A,\phi)}(\xi)=\big{(}[\xi,(d_{A}\phi)^{-}],[\xi,(F_{A}+\phi\wedge\phi)^{+}]\big{)}.$$
Thus the pair $(A,\phi)$ satisfies Kapustin$\en$Witten equations if and only if $d^{1}_{(A,\phi)}\circ d^{0}_{(A,\phi)}=0$. 
We can define the homology groups: $H^{0}_{(A,\phi)}=\ker d^{0}_{(A,\phi)}$, $H^{1}_{(A,\phi)}=\frac{\ker d^{1}_{(A,\phi)}}{\rm{Im}d^{0}_{(A,\phi)}}$ and $H^{2}_{(A,\phi)}=\rm{Coker}d^{1}_{(A,\phi)}$.
We denote the isotropy group of the pair $(A,\phi)$ as $\Ga(A,\phi)=\{u \in\mathcal{G}_{P}|u^{\ast}(A,\phi)=(A,\phi)\}$. 
Recall that $H^{0}_{(A,\phi)}$ is the Lie algebra of the stabilizer of $(A,\phi)$ and $H^{1}_{(A,\phi)}$ is the formal tangent space. 
The dual complex is:
$$0\ra\Om^{2,-}(X,\mathfrak{g}_{P})\oplus\Om^{2,+}(X,\mathfrak{g}_{P})\xrightarrow{d^{1,\ast}_{(A,\phi)}}\Om^{1}(X,\mathfrak{g}_{P})
\oplus\Om^{1}(X,\mathfrak{g}_{P})\xrightarrow{d^{0,\ast}_{(A,\phi)}}\Om^{0}(X,\mathfrak{g}_{P})\ra 0.$$
where
$$d^{1,\ast}_{(A,\phi)}(a',b')=(d_{A}^{\ast}b'-\ast[\phi,a'],d^{\ast}_{A}a'+\ast[\phi,b']),$$
$$d^{0,\ast}_{(A,\phi)}(a,b)=(-d^{\ast}_{A}a+\ast[\phi,\ast b]).$$
\begin{prop}\label{26}
The map $KW(A,\phi)$ has an exact quadratic expansion given by
$$KW(A+a,\phi+b)=KW(A,\phi)+d^{1}_{A,\phi}(a,b)+\{(a,b),(a,b)\},$$
where $\{(a,b),(a,b)\}$ is the symmetric quadratic form given by
$$\{(a,b),(a,b)\}:=\big{(}[a,b]^{-},(a\wedge a+[b,b])^{+}\big{)}.$$
\end{prop} 
Given fixed $(A_{0},\phi_{0})$, we look for solutions to the inhomogeneous equation $KW(A_{0}+a,\phi_{0}+b)=\psi_{0}$. By Proposition \ref{26}, this equation is equivalent to
\begin{equation}
d^{1}_{A_{0},\phi_{0}}+\{(a,b),(a,b)\}=\psi_{0}-KW(A_{0},\phi_{0}).
\end{equation}
To make this equation elliptic, it's nature to impose the gauge-fixing condition
$$d^{0,\ast}_{(A_{0},\phi_{0})}(a,b)=\zeta.$$
If we define
\begin{equation}\nonumber
\begin{split}
&\D_{(A_{0},\phi_{0})}:=d^{0,\ast}_{(A_{0},\phi_{0})}+d^{1}_{(A_{0},\phi_{0})},\\
&\psi=\psi_{0}-KW(A_{0},\phi_{0}).\\
\end{split}
\end{equation}
then the elliptic system can be rewritten as
\begin{equation}\label{27}
\D_{(A_{0},\phi_{0})}+\{(a,b),(a,b)\}=(\zeta,\psi).
\end{equation}
Local interior estimates for the elliptic system (\ref{27}) are consider in \cite{FL2} equation 3.2 in the context of $PU(2)$ monopoles. This is also considered similarly for the Vafa-Witten equations in \cite{BM}.
\subsection{An inequality for the connections near an ASD connection}
Let $A_{\infty}$ be a fixing ASD connection on $P$,\ any connection $A$ can be written uniquely as
$$A=A_{\infty}+a\ with\ a\in\Om^{1}(X,\mathfrak{g}_{P}).$$
In \cite{Don}, Donaldson proved that the connection $A$ can be written as
$$A=\tilde{A}_{\infty}+d^{+,\ast}_{A_{\infty}}u,$$
where $\tilde{A}_{\infty}$ is also an ASD connection and $u\in\Om^{2,+}(X,\mathfrak{g}_{P})$,\ i.e.,\ the connection $A$ satisfies
\begin{equation}\label{AE11}
-d^{+}_{A_{\infty}}d^{+,\ast}_{A_{\infty}}u+(d^{+,\ast}_{A_{\infty}}u\wedge d^{+,\ast}_{A_{\infty}}u)^{+}-([a\wedge d^{+,\ast}_{A_{\infty}}u])^{+}+F^{+}_{A}=0
\end{equation}
when the connection $A_{\infty}$ is $regular$ and $a$ is small enough in $L^{2}_{1}$-norm. The operator $d_{A}^{+}d^{+,\ast}_{A}$ is an elliptic self-adjoint operator on the space of $L^{2}$ sections of $\Om^{2,+}TX\otimes\mathfrak{g}_{P}$.
It is a standard result that the spectrum of $d_{A}^{+}d_{A}^{+,\ast}$ is discrete, and the lowest eigenvalue is nonnegative. 
\begin{defi}\label{D3}
For $A_{\infty}\in M_{ASD}$,\ define
\begin{equation*}
\mu(A_{\infty}):=\inf_{v\in\Om^{2,+}(X,\mathfrak{g}_{P})\backslash\{0\}}\frac{\|d^{+,\ast}_{A_{\infty}}v\|^{2}}{\|v\|^{2}}.
\end{equation*}
is the lowest eigenvalue of $d^{+}_{A_{\infty}}d^{+,\ast}_{A_{\infty}}$.
\end{defi}
One can see that an ASD connection $A_{\infty}$ is $regular$, i.e., $\mu(A_{\infty})>0$. The Sobolev norms $L^{p}_{k,A}$,\ where $1\leq p<\infty$ and $k$ is an integer, with respect to the connections defined as:
\begin{equation}\nonumber
\|u\|_{L^{p}_{k,A}(X)}:=\big{(}\sum_{j=0}^{k}\int_{X}|\na^{j}_{A}u|^{p}dvol_{g}\big{)}^{1/p}, \forall u\in L^{p}_{k,A}(X,\mathfrak{g}_{P}),
\end{equation}
where $\na^{j}_{A}:=\na_{A}\circ\ldots\circ\na_{A}$ (repeated $j$ times for $j\geq0$).
\begin{thm}\label{AT1}(\cite[Proposition 22]{Don})
Let $X$ be a closed, four-dimensional, smooth Riemannian manifold with a smooth Riemannian metric, $G$ be a compact Lie group, $P$ be a smooth principal $G$-bundle over $X$. If there is a $C^{\infty}$ ASD connection $A_{\infty}$ on $P$ that is $regular$, then there is constant $\sigma=\sigma(\mu(A_{\infty}),g,G)\in(0,1]$ with the following significance. If $A:=A_{\infty}+b$ is a smooth connection on $P$obeying
$$\|\na_{A}b\|_{L^{2}(X)}+\|b\|_{L^{2}(X)}\|F_{A_{\infty}}\|_{L^{4}(X)}\leq\sigma,$$
then there exist a solution $a:=d^{+,\ast}_{A_{\infty}}u\in\Om^{1}(X,\mathfrak{g}_{P})$ where $u\in\Om^{2,+}(X,\mathfrak{g}_{P})$ to Equation (\ref{AE11}).\ In fact,\ the connection $\tilde{A}_{\infty}:=A-a$ is an anti-self-dual connection on $P$. Further, there exist a constant $C=C(\mu(A_{\infty}),g,G)\in(0,\infty)$ such that
$$\|\na_{A}a\|_{L^{2}(X)}\leq C\|F^{+}_{A}\|_{L^{2}(X)}+\|F_{A_{\infty}}\|_{L^{4}}\|F^{+}_{A}\|_{L^{4/3}},$$
$$\|a\|_{L^{2}(X)}\leq C\|F^{+}_{A}\|_{L^{4/3}(X)}.$$
\end{thm}
This theorem \ref{AT1} by follows the method of proof of \cite[Theorem 2.2 ]{T4} applied to equation $F^{+}(A+d^{+,\ast}_{A_{\infty}}u)=0$.

\section{Non-existence solutions on a neighborhood of a $generic$ ASD connection}
\subsection{Decoupled Kapustin$\en$Witten equations}
\begin{defi}\label{D2.2}
Let $G$ be a compact Lie group, $P$ be a principal $G$-bundle over a closed, smooth four-manifold $X$ with a smooth Riemannian
metric $g$. The $decoupled$ Kapustin$\en$Witten equations on $P$ over $X$ are the equations require a pair $(A,\phi)\in\mathcal{A}_{P}\times\Om^{1}(X,\mathfrak{g}_{P}
)$ satisfies
\begin{equation*}
\begin{split}
&F^{+}_{A}=0,\  (\phi\wedge\phi)^{+}=0,\\ 
&d^{\ast}_{A}\phi=0,\ d_{A}\phi)^{-}=0.\\
\end{split}
\end{equation*}
\end{defi}
\begin{lem}
If a pair $(A,\phi)$ is a smooth solution of  $decoupled$ Kapustin$\en$Witten equations, then the extra field $\phi$ also obeys
$$\phi\wedge\phi=0,\ d_{A}\phi=0.$$
\end{lem}
\begin{proof}
At first, we observe that $\int_{X}tr(\phi\wedge\phi\wedge\phi\wedge\phi)=0$, hence 
$$\|\phi\wedge\phi\|^{2}_{L^{2}(X)}=2\|(\phi\wedge\phi)^{+}\|^{2}_{L^{2}(X)}=0,$$
i.e., $\phi\wedge\phi=0$. Following the identity $(d_{A}\phi)^{-}=0$, i.e., $d_{A}\phi=\ast d_{A}\phi$, we obtain that $$d^{\ast}_{A}d_{A}\phi=-\ast[F_{A},\phi].$$ 
Hence take the $L^{2}$ inner of above identity with $\phi$ and integrable by parts, we obtain
\begin{equation*}
\|d_{A}\phi\|^{2}_{L^{2}(X)}=-\int_{X}tr([F_{A},\phi)\wedge\phi]=
\int_{X}tr(F_{A}\wedge[\phi\wedge\phi])=0,
\end{equation*}
i.e., $d_{A}\phi=0$. We complete the proof of this lemma.
\end{proof}
We recall a vanishing theorem on the extra fields of $decoupled$ Kapustin$\en$Witten equations. The prove is similar to Vafa$\en$Witten equations \cite[Theorem 4.2.1]{BM}. At first, we recall a useful lemma proved by Donaldson \cite[Lemma 4.3.21]{DK}.
\begin{lem}\label{L1}
If $A$ is an irreducible $SU(2)$ or $SO(3)$ ASD connection on a bundle $P$ over a simply connected four-manifold $X$, then the restriction of $A$ to any non-empty open set in $X$ is also irreducible.
\end{lem}

\begin{thm}(\cite[Theorem 2.9] {HT} )\label{T2.6}
Let $X$ be a closed, simply-connected, smooth, oriented, Riemannian four-manifold, $P\rightarrow X$ be an $SU(2)$ or $SO(3)$ principal bundle. Suppose that $(A,\phi)$ is a solution of the $decoupled$ Kapustin$\en$Witten equations. If $A$ is irreducible, then the extra fields $\phi$ vanishes.
\end{thm}
\begin{proof}
We denote $Z^{c}$ by the complement of the zero of $\phi$. By unique continuation of the elliptic equation $(d_{A}+d^{\ast}_{A})\phi=0$, $Z^{c}$ is either empty or dense. Since $\phi\wedge\phi=0$, $\phi$ has at most rank one. The Lie algebra of $SU(2)$ or $SO(3)$ is three-dimensional, with basis $\{\sigma^{i}\}_{i=1,2,3}$ and Lie brackets
$\{\sigma^{i},\sigma^{j}\}=2\varepsilon_{ijk}\sigma^{k}$. In a local coordinate, we can set $\phi=\sum_{i=1}^{3}\phi_{i}\sigma^{i}$, where $\phi_{i}\in\Om^{1}(X)$. We then have
$$0=\phi\wedge\phi=2(\phi_{1}\wedge\phi_{2})\sigma^{3}+2(\phi_{3}\wedge\phi_{1})\sigma^{2}+2(\phi_{2}\wedge\phi_{3})\sigma^{1},$$
i.e.,
\begin{equation}\label{90}
0=\phi_{1}\wedge\phi_{2}=\phi_{3}\wedge\phi_{1}=\phi_{2}\wedge\phi_{3}.
\end{equation}
On $Z^{c}$,\ $\phi$ is non-zero,\ then  without loss of generality we can assume that $\phi_{1}$ is non-zero.\ From (\ref{90}),\ there exist functions $\mu$ and $\nu$ such that
$$\phi_{2}=\mu\phi_{1}\ and\ \phi_{3}=\nu\phi_{1}.$$
Hence we would re-written $\phi$ to
\begin{equation*}
\phi=\phi_{1}(\sigma^{1}+\mu\sigma^{2}+\nu\sigma^{3})=\phi_{1}(1+\mu^{2}+\nu^{2})^{1/2}(\frac{\sigma^{1}+\mu\sigma^{2}+\nu\sigma^{3}}{\sqrt{1+\mu^{2}+\nu^{2}}}).
\end{equation*}
Then on $Z^{c}$ write $\phi=\xi\otimes\w$ for $\xi\in\Om^{0}(Z^{c},\mathfrak{g}_{P})$ with $\langle\xi,\xi\rangle=1$,\ and $\w\in\Om^{1}(Z^{c})$.\ We compute
$$0=d_{A}(\xi\otimes\w)=d_{A}\xi\wedge\w-\xi\otimes d\w,$$
$$0=d_{A}\ast(\xi\otimes\w)=d_{A}\xi\wedge\ast\w-\xi\otimes d\ast\w.$$
Taking the inner product with $\xi$ and using the consequence of $\langle\xi,\xi\rangle=1$ that $\langle\xi,d_{A}\xi\rangle=0$, we get $d\w=d^{\ast}\w=0$. It follows that $d_{A}\xi\wedge\w=0$ and $d_{A}\xi\wedge\ast\w=0$. Since $\w$ is nowhere zero along $Z^{c}$, we must have $d_{A}\xi=0$ along $Z^{c}$. Therefore, $A$ is reducible along $Z^{c}$. However according to  Lemma \ref{L1}, $A$ is irreducible along $Z^{c}$. This is a contradiction unless $Z^{c}$ is empty. Therefore $Z=X$, so $\phi$  identically zero.
\end{proof} 
\subsection{A bounded property of extra fields}
In this section, we will proved that the extra fields have a lower positive bounded if the connections on the neighborhood of a $regular$ ASD connection. The result has a direct proof by using the Kuranshi model for the moduli space of Kapustin$\en$Witten solutions around the pair $(A_{\infty},0)$, where $A_{\infty}$ is an ASD connection. Our proof employs a Weizenb\"{o}ck formula, the idea could extends to the cases of Kapustin$\en$Witten equations and Vafa$\en$Witten equations on K\"{a}hler surface. We need recall a bounded on $\|\phi\|_{L^{\infty}}$ in terms of $\|\phi\|_{L^{2}}$. The technique is analogy to Vafa$\en$Witten equations case \cite{BM}.
\begin{thm}(\cite[Theorem 2.4]{HT} )\label{T2.1}
If the pair $(A,\phi)$ is a smooth solution of Kapustin$\en$Witten equations over a closed, smooth, Riemannian four-manifold, then
$$\|\phi\|_{L^{\infty}(X)}\leq C\|\phi\|_{L^{2}(X)},$$
where $C=C(g)$ is a positive constant only depends on metric $g$.
\end{thm}
If we suppose that the connection $[A]$ is in a neighborhood of a $regular$ ASD connection $[A_{\infty}]$, then the Theorem \ref{AT1} which provides existence of an other ASD connection $\tilde{A}_{\infty}$ on $P$ and a Sobolev norm estimate for the distance between $A$ and $\tilde{A}_{\infty}$. Following the idea in \cite{HT}, we then have
\begin{prop}\label{P3.9}
Let $X$ be a closed, oriented, smooth, four-dimensional Riemannian manifold with Riemannian metric $g$, $P\rightarrow X$ be a principal $G$-bundle with $G$ being a compact Lie group with $p_{1}(P)$ negative. Suppose that $A_{\infty}$ is an $regular$ ASD connection on $P$, i.e., $\ker {d^{+}_{A_{\infty}}|_{\Om^{2,+}(X,\mathfrak{g}_{P})}}=0$, then there are positive constants $\de=\de(g,A_{\infty})\in(0,1)$, $C=C(g,A_{0})\in[1,\infty)$ with following significance. If $(A,\phi)$ is a smooth solution of Kapustin$\en$Witten equations and the connection $A$ obeys
$$dist(A,A_{\infty})\leq\de$$
then either the extra field obeying
$$\|\phi\|_{L^{2}}\geq C$$ 
or $A$ is anti-self-dual with  respect to $g$.
\end{prop}
\begin{proof}
Following Theorem \ref{AT1}, it implies that for a suitable constant $\delta$, the connection $A$ can be written as $A=\tilde{A}_{\infty}+a$, where $\tilde{A}_{\infty}$ is also an ASD connection. Following the idea in \cite{HT}, we have two integrable inequalities: 
$$	\|\na_{A}\phi\|^{2}_{L^{2}(X)}+\langle Ric\circ\phi,\phi\rangle_{L^{2}(X)}+2\|F^{+}_{A}\|^{2}_{L^{2}(X)}=0,$$
$$ \|\na_{\tilde{A}_{\infty}}\phi\|^{2}_{L^{2}(X)}+\langle Ric\circ\phi,\phi\rangle_{L^{2}(X)}\geq0.$$
Combining the preceding inequalities gives
\begin{equation}\nonumber
\begin{split}	0&\leq\|\na_{A_{0}}\phi\|^{2}_{L^{2}(X)}+\langle Ric\circ\phi,\phi\rangle_{L^{2}(X)}\\
&\leq\|\na_{A}\phi\|^{2}_{L^{2}(X)}+\langle Ric\circ\phi,\phi\rangle_{L^{2}(X)}+\|\na_{A}\phi-\na_{\tilde{A}_{\infty}}\phi\|_{L^{2}(X)}^{2}\\
&\leq c\|a\|_{L^{2}(X)}\|\phi\|_{L^{\infty}(X)}-2\|F^{+}_{A}\|^{2}_{L^{2}(X)}\\
&\leq (c\|\phi\|_{L^{2}(X)}-2)\|F_{A}^{+}\|_{L^{2}(X)}.\\
\end{split}
\end{equation}
If $\|F^{+}_{A}\|_{L^{2}(X)}$ is non-zero, thus $\|\phi\|^{2}_{L^{2}(X)}\geq2/c$. We complete the proof of this proposition.
\end{proof}
\subsection{Uhlenbeck type compactness of Kapustin$\en$Witten equations}
We now recall the work by Taubes on Uhlenbeck style  compactness theorem of Kapustin$\en$Witten equations on closed four-dimensional manifolds \cite{T1}. In \cite{TY2}, Tanaka also briefly describe the compactness theorem. Let $X$ be a closed, oriented, smooth Riemannian four-manifold with Riemannian metric $g$, $P\rightarrow X$ be a principal $G$-bundle over $X$ with $G$ being $SU(2)$ or $SO(3)$. We consider a sequence solutions $\{(A_{i},\phi_{i})\}$ of the Kapustin$\en$Witten equations. We put $r_{i}:=\|\phi_{i}\|_{L^{2}(X)}$. In the case that $\{r_{i}\}$ has a bounded subsequence, the Uhlenbeck compactness theorem with the Weizenb\"{o}ck formula deduces the following.
\begin{thm}(\cite[Proposition 2.1]{TY2}  and \cite{T1})\label{51}
If $\{r_{i}\}$ has a bounded subsequence, then \\
(1) there exist a principal bundle $P_{\De}\rightarrow X$ and a pair $(A_{\De},\phi_{\De})\in\mathcal{A}_{P_{\De}}\times \Om^{1}(X,\mathfrak{g}_{P_{\De}})$ obeys the Kapustin-Witten equations; \\
(2) a finite set $\Sigma\subset X$ of points, a subsequence $\Xi\in\N$ and a sequence $\{g_{i}\}_{i\in\Xi}$ of automorphisms of $P_{\De}|_{X-\Sigma}$ such that $\{(g_{i}^{\ast}(A_{i}),g^{\ast}_{i}(\phi_{i}))\}_{i\in\Xi}$ converges to $(A_{\De},\phi_{\De})$ in the $C^{\infty}$ topology on compact subsets in $X-\Sigma$.
\end{thm}
Analysis for the case that $\{r_{i}\}$ has no bounded subsequences was the bulk of \cite[Theorem 1.1]{T1} . 
\begin{thm}(\cite[Proposition 2.2]{TY2}  and Theorem 2.3)\label{T4.6}
If $\{r_{i}\}$ has no bounded subsequence. There exists in this case the following data:\\
(1) A finite set $\Theta\subset X$ and a closed, nowhere dense set $Z\subset X-\Theta$,\\
(2) a real line bundle $\mathcal{I}\rightarrow X-(Z\cup\Theta)$,\\
(3) a section $v$ of $\mathcal{I}\otimes T^{\ast}X|_{X-(Z\cup\Theta})$ with $dv=d^{\ast}v=0$, the norm of $v$ extends over the whole of $X$ as a bounded $L^{2}_{1}$ function. In addition,\\
a) The extension of $|v|$ is continuous on $X-\Theta$ and its zero locus is the set $Z$.\\
b) Let $U$ denote an open set in $X-\Theta$ with compact closure. The function $|v|$ is H\"{o}lder continuous on $U$ with H\"{o}lder exponent that is independent of $U$ and of the original sequence $\{(A_{i},\phi_{i})\}_{i=1,2,\ldots}$.\\
c) If $p$ is any given point in $X$, then the function $dist(\cdot,p)^{-1}|\na v|$ extends to the whole of $X$ as an $L^{2}_{1}$-function.\\
(4) A principal $G$-bundle $P_{\De}$ over $X-(Z\cup\Theta)$ and a connection $A_{\De}$ of $P_{\De}$ with $d_{A_{\De}}^{\ast}F_{A_{\De}}=0$.\\
(5) An isometric $A_{\De}$ covariantly constant homorphism $\sigma_{\De}:\mathcal{I}\rightarrow\mathfrak{g}_{P}$.\\
In addition, there exist a subsequence $\Xi\subset\N$ and a sequence of automorphisms $g_{i}:P_{\De}\rightarrow P|_{X-(Z\cup\Theta)}$ such that\\
(i) $\{g_{i}^{\ast}(A_{i})\}_{i\in\Xi}$ converges to $A_{\De}$ in the $L^{2}_{1}$ topology on compact subset in $X-(Z\cup\Theta)$ and\\
(ii) The sequence $\{r^{-1}g^{\ast}_{i}(\phi_{i})\}$ converges to $v\otimes\sigma_{\De}$ in $L^{2}_{1}$ topology on compact subset in $X-(Z\cup\Theta)$ and $C^{0}$-topology on $X-\Theta$. Moreover, the sequence $\{r_{i}^{-1}\phi_{i}\}_{i\in\Xi}$ is pointwise bounded and it converges in the $L^{2}_{1}$ topology on $X$ to $|v|$.
\end{thm}
As was the case for Uhlenbeck theorem, the number of elements in $\Theta$ only dependent on the first Pontrjagin class $p_{1}(P)$. In particular,  $\Theta$ is empty if the first Pontrjagin class of $p_{1}(P)=0$. As for the structure of the above $Z$, Taubes proved that $Z$ has the Hausdorff dimension at most 2, see \cite{T1}.

For a sequence of connections $\{A_{i}\}$ on $P$, Tanaka defined a set $S(\{A_{i}\})$ as follows:
\begin{equation}\label{E11}
S(\{A_{i}\}):=\bigcap_{r\in(0,\rho)}\{x\in X:\liminf_{i\rightarrow\infty}\int_{B_{r}(x)}|F_{A_{i}}|^{2}\geq\varepsilon_{0}\},
\end{equation}
where $\varepsilon_{0}$ is a positive constant is defined as in \cite{Tanaka}. This set $S(\{A_{i}\})$ describes the singular set of a sequence of connections $\{A_{i}\}$. In \cite{Tanaka}, Tanaka observed that under the particular circumstance where the connections are non-concentrating and the limiting connection is non-locally reducible, one obtains an $L^{2}$-bounded on the extra fields. We extend the idea used by Tanaka for Vafa$\en$Witten equations (see \cite[Theorem 1.3]{Tanaka}) to the Kapustin$\en$Witten equations case. 
\begin{thm}\label{T2}
Let $X$ be a closed, oriented,  4-dimensional manifold with a smooth Riemannian metric $g$, $P\rightarrow X$ be a principal $G$-bundle with $G$ being $SU(2)$ or $SO(3)$. If $\{(A_{i},\phi_{i})\}$ is a sequence of smooth solutions of Kapustin$\en$Witten equations with $S(\{A_{i}\})$ being empty,	then there exists subsequence $\Xi\subset\N$ and a sequence of gauge transformations $\{g_{i}\}_{i\in\Xi}$ such that $\{g_{i}(A_{i})\}_{i\in\Xi}$ converges weakly in the $L^{2}_{1}$-topology. If the limit is not locally reducible, then there exists a positive number $C$ such that $\int_{X}|\phi_{i}|^{2}\leq C$ for all $i\in\Xi$, and $\{(g_{i}^{\ast}(A_{i}),g_{i}^{\ast}(\phi_{i}))\}_{i\in\Xi}$ converges in the $C^{\infty}$-topology to a pair that obeys the Kapustin$\en$Witten equations.
\end{thm}
We then have a useful lemma as follows.
\begin{lem}\label{L7}
Let $X$ be a closed, oriented, simply-connected, 4-dimensional manifold with a smooth Riemannian metric $g$, $P\rightarrow X$ be a principal $G$-bundle with $G$ being $SU(2)$ or $SO(3)$. If $\{(A_{i},\phi_{i})\}$ is a sequence of smooth solutions of Kapustin$\en$Witten equations such that
$$dist(A_{i},A_{\infty})\rightarrow0,\ as\ i\rightarrow\infty,$$
where  $A_{\infty}$ is an irreducible ASD connection on $P$, then there exist a subsequence $\Xi\subset\mathbb{N}$ such that the sequence $\{r_{i}:=\|\phi_{i}\|_{L^{2}(X)}\}_{i\in\Xi}$ is a bounded subsequence.
\end{lem}
\begin{proof}
For the sequence of $\{A_{i}\}_{i\in\N}$ on $P$, one can see $S(\{A_{i}\})=\varnothing$, where $S(A_{i})$ defined as in equation (\ref{E11}). Following \cite[Theorem 1.3]{Tanaka2015} , there exist a subsequence $\Xi\subset\N$ and a sequence of gauge transformations $g^{\ast}_{i}(A_{i})_{i\in\Xi}$ converges weakly in the $L_{1}^{2}$-topology, we denote the limit connection by $\tilde{A}_{\infty}$. Under our assumption on the sequence $\{A_{i}\}_{i\in\N}$, there exist a gauge transformation $g_{\infty}$ such that $g_{\infty}^{\ast}(A_{\infty})=\tilde{A}_{\infty}$. Hence $\tilde{A}_{\infty}$ is also not locally reducible, since the locally reducible connection on simply-connected manifold is also reducible, see \cite[Appendix B]{Tanaka} . At last, following Theorem \ref{T2}, it implies that there eixst a positive number $C$ such that $\|\phi_{i}\|_{L^{2}(X)}\leq C$ for all $i\in\Xi$.
\end{proof}
\begin{proof}[\textbf{Proof of Theorem \ref{T1.1}}.] We suppose that the pair $(A,\phi)$ is a smooth solution of Kapustin$\en$Witten equations and $A$ is not ASD with respect to metric $g$. If we suppose that the constant $\de$ does not exist, we may choose a sequence $\{(A_{i},\phi_{i})\}_{i\in\N}$ of solutions of Kapustin$\en$Witten equations on $P$ such that $dist(A_{i},A_{\infty})\rightarrow0$ as $i\rightarrow\infty$, then following Proposition \ref{P3.9} and Lemma \ref{L7}, there exists a subsequence $\Xi\subset\mathbb{N}$ and two positive constants $C,c$,\ such that
$$c\leq\|\phi_{i}\|_{L^{2}(X)}\leq C.$$
Following the compactness Theorem \ref{51}, there exist a pair $(A_{\De},\phi_{\De})\in\mathcal{A}_{P_{\De}}\times \Om^{1}(X,\mathfrak{g}_{P_{\De}})$ that obeys the Kapustin$\en$Witten equations and there has a subsequence $\Xi'\subset \Xi$ and a sequence $\{g_{i}\}_{i\in\Xi'}$ of automorphisms of $P_{\De}$ such that $\{(g_{i}^{\ast}(A_{i}),g^{\ast}_{i}(\phi_{i}))\}_{i\in\Xi'}$ converges to $(A_{\De},\phi_{\De})$ in the $C^{\infty}$ topology on $X$. Thus the extra field $\phi_{\De}$ also has a lower positive bounded:
$$\|\phi_{\De}\|_{L^{2}(X)}\geq\liminf\|\phi_{i}\|_{L^{2}(X)}\geq c.$$
But on the other hand, under our initial assumption regarding the sequence $\{A_{i}\}_{i\in\N}$, we have $[A_{\De}]\equiv [A_{\infty}]$. Thus the connection $A_{\De}$ is also irreducible, following the vanishing Theorem \ref{T2.6}, it implies that the extra field $\phi_{\De}=0$. This contradicts $\|\phi_{\De}\|_{L^{2}(X)}$ has a uniform positive lower bound. The preceding argument shows that the desired constant $\de$ exists.
\end{proof}
\subsection{A bounded property of curvatures}
In this section, we extend the proof method of Theorem \ref{T1.1} to global case. We suppose that the connections $[A_{\infty}]\in\bar{M}_{ASD}$ are all $regular$, i.e., $\ker {d^{+}_{A_{\infty}}|_{\Om^{2,+}(X,\mathfrak{g}_{P})}}=0$. Following the result in \cite[Section 3]{PF2}, one know that the connections on 
$$\mathcal{B}_{\varepsilon}(P,g):=\{A\in\mathcal{A}_{P}:\|F^{+}_{A}\|_{L^{2}(X)}\leq\varepsilon\}$$
also obey  $\ker{d^{+}_{A_{\infty}}|_{\Om^{2,+}(X,\mathfrak{g}_{P})}}=0$ for suitable constant $\varepsilon$.
\begin{prop}\label{P1}
Let $G$ be a compact Lie group, $P\rightarrow X$ be a principal $G$-bundle over a closed, connected, four-dimensional manifold $X$ with a smooth Riemannian metric $g$. Suppose that the connections on $\bar{M}_{ASD}$ are all $regular$, i.e., $\ker {d^{+}_{A_{\infty}}|_{\Om^{2,+}(X,\mathfrak{g}_{P})}}=0$ for any connection $[A_{\infty}]\in\bar{M}_{ASD}$, then there are  positive constants $\varepsilon=\varepsilon(P,g)$ and $\mu=\mu(P,g)$ such that
\begin{equation}\nonumber
\mu(A)\geq\mu,\  [A]\in\mathfrak{B}_{\varepsilon}(P,g).
\end{equation}
where $\mu(A)$ is as in Definition \ref{D3}.
\end{prop}
We then recall a result of extra fields $\phi$ which proved by author \cite{HT}.
\begin{thm}\label{T3.11}
Assume the hypotheses of Proposition \ref{P1}. Suppose that the connections on $\bar{M}_{ASD}$ are all $regular$, i.e., $\ker {d^{+}_{A_{\infty}}|_{\Om^{2,+}(X,\mathfrak{g}_{P})}}=0$ for any connection $[A_{\infty}]\in\bar{M}_{ASD}$, then there is a positive constant $c=c(P,g)$ with following significance. If $(A,\phi)$ is a smooth solution of Kapustin$\en$Witten equations, then the extra filed $\phi$ obeys
$$\|\phi\|_{L^{2}(X)}\geq c,$$
unless $A$ is anti-self-dual with respect to $g$.
\end{thm}
Now, we begin to consider a sequence smooth solutions $\{(A_{i},\phi_{i})\}_{i\in\mathbb{N}}$ of Kapustin$\en$Witten equations. If we suppose that $\|\phi_{i}\|_{L^{2}(X)}$ has no bounded subsequence, following the compactness theorem \ref{T4.6} proved by Taubes, we only know that the connection $A_{\De}$ has harmonic curvature. Moreover if we suppose that $F_{A_{i}}^{+}$ converge to zero in $L^{2}$-topology, then we can claim that $A_{\De}$ is an anti-self-dual connection on the complement of $Z\cup\Theta\cup\Sigma$.
\begin{cor}\label{C3.10}
Let $\{(A_{i},\phi_{i})\}_{i\in\N}$ be a sequence smooth solutions of Kapustin$\en$Witten equations, set $r_{i}:=\|\phi_{i}\|_{L^{2}(X)}$. Suppose that $\{F^{+}_{A_{i}}\}_{i\in\N}$ converge to zero in $L^{2}$-topology and the sequence $\{r_{i}\}_{i=1,2,\ldots}$ has no bounded subsequence.\ Let $Z$,\ $\Theta$ and $\mathcal{I}$ be as described in Theorem \ref{T4.6},\ so that $\sigma_{\De}$ and $A_{\De}$ are defined over $X-(Z\cup\Theta\cup\Sigma)$.\ Then the connection $A_{\De}$ is anti-self-dual connection on $P_{\De}$.
\end{cor}
Before the proof of Corollary \ref{C3.10},\ we should recall a key  a compactness theorem due to Sedlacek, see \cite[Theorem 4.3]{Sedlacek} or \cite[Theorem 35.17]{PF1} .
\begin{thm}\label{T4.4}
Let $G$ be a compact Lie group and P a principal $G$-bundle over a closed, smooth, oriented, four-dimensional Riemannian manifold $X$ with a Riemannian metric $g$. If $\{A_{i}\}_{i\in\mathbb{N}}$ is a sequence $C^{\infty}$ connection on $P$ and the sequence $\{F^{+}_{A_{i}}\}_{i\in\N}$ converges to zero in $L^{2}$-topology, then there exists\\
(1) An integer $L$ and a finite set of points, $\Sigma=\{x_{1},\ldots,x_{L}\}\subset X$, $\Sigma$ maybe empty,\\
(2) A smooth anti-self-dual $A_{\infty}$ on a principal $G$-bundle $P_{\infty}$ over $X$ with $\eta(P_{\infty})=\eta(P)$,\\
(3) A subsequence $\Xi\subset\N$, we also denote by $\{A_{i}\}_{i\in\Xi}$, a sequence gauge transformation $\{g_{i}\}_{i\in\Xi}$ such that, $g^{\ast}_{i}(A_{i})$ weakly converges to $A_{\infty}$ in $L^{2}_{1}$ on $X-\Sigma$, and $g_{i}^{\ast}(F_{A_{i}})$ weakly converges to $F_{A_{\infty}}$ in $L^{2}$ on $X-\Sigma$.
\end{thm}
\begin{proof}[\textbf{Proof of Corollary \ref{C3.10}}.] We can apply Theorem \ref{T4.4} to the sequence $\{A_{i}\}_{i\in\N}$. This yields a subsequence $\Xi\subset\N$, a sequence of gauge transformations $\{g_{i}\}_{i\in\Xi}$ and a anti-self-dual connection $A_{\infty}$ on a principal $G$-bundle $P_{\infty}$ which is the weakly $L^{2}_{1}$ limit of
$\{g_{i}^{\ast}(A_{i})\}_{i\in\Xi}$ over $X-\Sigma$,\ $\Sigma$ is the set of some points on $X$. If the sequence $r_{i}$ don't has a bounded subsequence. Apply Theorem \ref{T4.6} to the subsequence $\{(A_{i},\phi_{i})\}_{i\in\Xi}$. Recall from Theorem \ref{T4.6}  that $A_{\De}$ is the limit over compact subset of $X-(Z\cup\Theta)$ of gauge transformations of $\{A_{i}\}_{i\in\Xi}$. In particular, both $A_{\infty}$ and $A_{\De}$ are weakly $L^{2}_{1}$ limits over $X-(Z\cup\Theta\cup\Sigma)$ of gauge-equivalent connections. Since weakly $L^{2}_{1}$ limits preserve $L^{2}_{2}$ gauge equivalence, it follows that there exists a Sobolev-class $L^{2}_{2}$ gauge transformation $g_{\De}$ such that $g_{\De}^{\ast}(A_{\De})=A_{\infty}$. Thus $A_{\De}$ is anti-self-dual on the complement of $Z\cup\Theta\cup\Sigma$.
\end{proof}
The  result due to Tanaka \cite[Theorem 1.3]{Tanaka2015}  for the sequence $\{A_{i}\}_{i\in\N}$ obeys $\|F_{A_{i}}^{+}\|_{L^{2}(X)}$ converges to zero in $L^{2}$-topology  will invalid, since $S({A_{i}})$ may be not empty. But we could prove a compactness theorem for a sequence solutions of Kapustin$\en$Witten equations by using the Lemma \ref{L1}.
\begin{prop}\label{P3.11}
Let $X$ be a closed, oriented, simply-connected, four-dimensional manifold with a smooth Riemannain metric $g$, $P\rightarrow X$ be a principal $SU(2)$ or $SO(3)$-bundle with $p_{1}(P)$ negative. Suppose that the connections on $\bar{M}_{ASD}(P,g)$ are irreducible. If $\{(A_{i},\phi_{i})\}_{i\in\N}$ is a sequence smooth solutions of Kapustin$\en$Witten equations and the curvatures 
$F^{+}_{A_{i}}$ converge to zero in $L^{2}$-topology, then\\
(1) An integer $L$ and a finite set of points, $\Sigma=\{x_{1},\ldots,x_{L}\}\subset X$,;\\
(2) A subsequence $\Xi\subset\N$, we also denote by $\{A_{i}\}$, a sequence gauge transformatin $\{g_{i}\}_{i\in\Xi}$ such that, $(g^{\ast}_{i}(A_{i}),g_{i}^{\ast}(\phi_{i}))$  converges to $(A_{\infty},0)$ in $C^{\infty}$ on $X-\Sigma$, where $A_{\infty}$ is an anti-self-dual connection on a principal $P_{\infty}$.
\end{prop}
\begin{proof}
At first, we claim that the sequence $r_{i}:=\|\phi_{i}\|_{L^{2}(X)}$ has a bounded subsequence. If not, the sequence $r_{i}$ don't has a bounded subsequence. Following Theorem \ref{T4.6}, we have $Z,\Theta$ and $\sigma_{\De}$, $v$ which described in Theorem \ref{T4.6}. We define  $\sigma_{\infty}:=g_{\De}^{\ast}(\sigma_{\De})$ over $X-(Z\cup\Theta\cup\Sigma)$, where $g_{\De}$, $\Sigma$ are as described in the proof of Corollary \ref{C3.10}. Following Theorem \ref{T4.6}, we then have 
$$\na_{A_{\infty}}\sigma_{\infty}=\na_{A_{\De}}\sigma_{\De}=0\ on\ X-(Z\cup\Theta\cup\Sigma).$$
Thus we have a section $s:=v\otimes\sigma_{\infty}$ on $P_{\infty}|_{X-(Z\cup\Theta\cup\Sigma)}$ and one can see $v\otimes\sigma_{\infty}$ is non-zero all over $X-(Z\cup\Theta\cup\Sigma)$. We can rewrite $s$ to $s=\tilde{\sigma}\otimes\tilde{v}$, where $\tilde{\sigma}\in\Gamma(X-(Z\cup\Theta\cup\Sigma),\mathfrak{g}_{P_{\infty}})$ and $\tilde{v}\in\Om^{1}(X-(Z\cup\Theta\cup\Sigma))$. We also setting $\langle\tilde{\sigma},\tilde{\sigma}\rangle=1$, thus $\langle d_{A_{\infty}}\tilde{\sigma},\tilde{\sigma} \rangle=0$ along $X-(Z\cup\Theta\cup\Sigma)$. In a direct calculate, we have $$d_{A_{\infty}}\tilde{\sigma}\wedge\tilde{v}+\tilde{\sigma}\wedge d\tilde{v}=0,$$ thus $d\tilde{v}=0$. It follows that $d_{A_{\infty}}\tilde{\sigma}\wedge v=0$. Since $\tilde{v}$ is nowhere zero along $X-(Z\cup\Theta\cup\Sigma)$, we must have $d_{A_{\infty}}\tilde{\sigma}=0$. According to Lemma \ref{L1}, $A$ is irreducible along  a open set of $X-(Z\cup\Theta\cup\Sigma)$, then $\tilde{\sigma}=0$. This contradict  $s$ is non-zero on $X-(Z\cup\Theta\cup\Sigma)$. Hence we prove that the sequence $\{r_{i}\}_{i=1,2,\ldots}$ must has a bounded subsequence. Then following the compactness theorem \ref{51}, there exist a pair $(A_{\De},\phi_{\De})\in\mathcal{A}_{P_{\De}}\times \Om^{1}(X,\mathfrak{g}_{P_{\De}})$ that obeys the Kapustin$\en$Witten equations and there has a subsequence $\Xi'\subset \Xi$ and a sequence $\{g_{i}\}_{i\in\Xi'}$ of automorphisms of $P_{\De}$ such that $\{(g_{i}^{\ast}(A_{i}),g^{\ast}_{i}(\phi_{i}))\}_{i\in\Xi'}$ converges to $(A_{\De},\phi_{\De})$ in the $C^{\infty}$ topology on $X-\{x_{1},\cdots,x_{L}\}$. Under the assumption of $\{A_{i}\}_{i\in\N}$, the connection $A_{\De}$ is an irreducible anti-self-dual connection, then the vanish theorem ensures $\phi_{\De}=0$.
\end{proof}
\begin{proof}[\textbf{Proof of Theorem \ref{T1.2}}.] Now we begin to proof Theorem \ref{T1.2}. Suppose that the constant $\de$ does not exist. We may choose a sequence of solutions $\{(A_{i},\phi_{i})\}_{i\in\N}$ of Kapustin$\en$Witten equations such that $\|F^{+}_{A_{i}}\|_{L^{2}(X)}\rightarrow 0$. Following Proposition \ref{P3.11}, there exist a subsequence $\Xi\in\N$ and a sequence transformation $\{g_{i}\}_{i\in\Xi}$ such that 
$(g^{\ast}_{i}(A_{i}),g^{\ast}_{i}(\phi_{i}))\rightarrow (A_{\infty},0)$ in $C^{\infty}$ over $X-\Sigma$, $\Sigma:=\{x_{1},\cdots,x_{L}\}$. There also exist a positive constant $C$, such that $\|\phi_{i}\|_{L^{2}(X)}\leq C$. We then have
$$\|\phi_{i}\|_{L^{\infty}(X)}\leq c\|\phi_{i}\|_{L^{2}(X)}\leq cC.$$
where $c=c(g)$ is a positive constant. Therefore
$$\lim_{i\rightarrow\infty}\int_{X}|\phi_{i}|^{2}=\lim_{i\rightarrow\infty}\int_{X-\Sigma}|\phi_{i}|^{2}+\lim_{i\rightarrow\infty}\int_{\Sigma}|\phi_{i}|^{2}\leq cC\mu(\Sigma)=0.$$
This contradict $\|\phi_{i}\|_{L^{2}(X)}$ has a uniform positive lower bound, see Theorem \ref{T3.11}. The preceding argument shows that the desired constant $\de$ exists.

If we denote by $A_{\infty}$ an ASD connection on $P$, then the curvature $F_{A}$ of a connection $A:=A_{\infty}+a$ has an estimate
\begin{equation}\nonumber
\begin{split}
\|F^{+}_{A}\|_{L^{2}(X)}&=\|(d_{A_{\infty}}a+a\wedge a)^{+}\|_{L^{2}(X)}\\
&\leq \|d_{A_{\infty}}a\|_{L^{2}(X)}+\|a\wedge a\|_{L^{2}(X)}\\
&\leq C(\|\na_{A_{\infty}}a\|_{L^{2}(X)}+\|a\|^{2}_{L^{4}(X)})\\
&\leq C(\|a\|_{L^{2}_{1}(X)}+\|a\|^{2}_{L^{2}_{1}(X)}),\\
\end{split}
\end{equation}
where $C=C(g)$ is a positive constant. If $\|a\|_{L^{2}_{1}(X)}\leq 1$, then
$$\|F^{+}_{A}\|_{L^{2}(X)}\leq 2C\|a\|_{L^{2}_{1}(X)}.$$
We also suppose that $A$ is not anti-self-dual with respect to $g$, then
$$\|a\|_{L^{2}_{1}(X)}\geq\frac{\de}{2C},$$
where $\de$ is the positive constant as above. So we can set $\tilde{\de}:=\min\{1,\frac{\de}{2C}\}$, hence
$$dist(A,M_{ASD}):=\inf_{g\in\mathcal{G},A_{\infty}\in M_{ASD}}\|g^{\ast}(A)-A_{\infty}\|_{L^{2}_{1}(X)}\geq\tilde{\de}$$
unless $A$ is ASD with respect to metric $g$.
\end{proof}
\subsection{Some examples}
In this section, we give some conditions on the topology of manifold, the metric of manifold and the topology of principle bundle to ensure that the connections on $\bar{M}_{ASD}$ are all $generic$. For a compact four-manifold $X$, the compacitification $\bar{M}_{ASD}(P,g)$ of $M_{ASD}(P,g)$ contained in the disjoint union
\begin{equation}\label{E3.2}
\bar{M}_{ASD}(P,g)\subset\cup(M_{ASD}(P_{l},g)\times Sym^{l}(X)).
\end{equation}
We denote by $\eta(P)$ the element in $H^{2}(X,\mathbb{R})$ which defined as \cite[Definition 2.1]{Sedlacek}. Following \cite[Theorem 5.5]{Sedlacek} , every principal $G$-bundle, $M(P_{l},g)$ over $X$ appearing in (\ref{E3.2}) has the property that $\eta(P_{l})=\eta(P)$.
\begin{prop}\label{P4}\label{P2.6}
Let $X$ be a closed, oriented, simply-connected, four-dimensional manifold with a $generic$ Riemannain metric $g$, $P\rightarrow X$ be a principal $SU(2)$ or $SO(3)$-bundle with $p_{1}(P)$ negative. If $b^{+}(X)>0$, then the connection $[A]\in M_{ASD}$ is irreducible.
\end{prop}
\begin{proof}
If $G=SU(2)$ or $SO(3)$, $X$ is a simply-connected four-manifold and $b^{+}(X)>0$, then following \cite[Corollary 4.3.15]{DK} , the only reducible anti-self-dual connection on a principal $G$-bundle over $X$, is the product connection on the product bundle $P=X\times G$ if only if the anti-self-dual connection is flat connection, i.e., $p_{1}(P)=0$. Hence under our assumption on $X,P,G$, then the anti-self-dual connection must be irreducible.
\end{proof}
In Proposition \ref{P4} we mean by $generic$ $metric$ the metrics in a second category subset of the space of $C^{k}$ metrics for some fixed $k>2$ (\cite[Section 4]{DK} and \cite[Corollary 2]{PF2}). It may reassure the reader to know that for all practical purposes one can work with an open dense subset of the smooth metrics, or even real analytic metrics.
\begin{prop}\label{2.11}
Let $X$ be a closed, oriented, simply-connected, four-dimensional manifold with a $generic$ $Riemannain$ $metric$ $g$, $P\rightarrow X$ be a principal $SO(3)$-bundle with $p_{1}(P)$ negative.\ If $b^{+}(X)>0$ and the second Stiefel-Whitney class $w_{2}(P)\neq0$, then the connections $[A]\in\bar{M}_{ASD}$ are all $generic$.
\end{prop}
\begin{proof}
Under the hypothesis on this proposition, we know that the connections $[A]\in\bar{M}_{ASD}$ are all $regular$, i.e., $\mu(A)>0$ where $\mu(A)$ is as in Definition \ref{D3}, see \cite[Corollary 3.9]{PF2} . If $G=SO(3)$, following \cite[Theorem 2.4]{Sedlacek} , we have $\eta(P)=w_{2}(P)$. We then obtain that if  $w_{2}(P)$ is non-trivial, then every principal $G$-bundle, $M(P_{l},g)$, over $X$ appearing in (\ref{E3.2}) has the property that $w_{2}(P_{l})$ is non-trivial. We would claim that the ASD connections on $M(P_{l},g)$ are irreducible. Since an reducible ASD connection on $P_{l}$ ensures that the bundle $P_{l}$ is trivial bundle. This contradict $w_{2}(P_{l})\neq 0$. By the similar method in Proposition \ref{P2}, for any $[A]\in\bar{M}_{ASD}(P,g)$, we have $\la(A)>0$, where $\la(A)$ is as in Definition (\ref{D2}), i.e., $[A]$ is irreducible. We complete the proof of this proposition.
\end{proof}
Following Theorem \ref{T1.2}, we then have
\begin{cor}
Assume the hypothesis on Proposition \ref{2.11}. Suppose that $b^{+}(X)>0$ and the second Stiefel-Whitney class $w_{2}(P)\neq0$. There is a positive constant $\de=\de(P,g)$ with following significance. If the pair $(A,\phi)$ is a smooth solution of Kapustin$\en$Witten equations, then one of following must hold:\\
(1) $F^{+}_{A}=0$ and $\phi=0$;\\
(2) the pair $(A,\phi)$ satisfies
$$2\|\phi\|^{2}_{L^{2}}\geq\|F^{+}_{A}\|_{L^{2}(X)}\geq\de.$$
In particular, there is a positive constant $\tilde{\de}=\tilde{\de}(g,P)$ such that
$$dist(A,M_{ASD}):=\inf_{g\in\mathcal{G}_{P},A_{\infty}\in M_{ASD}}\|g^{\ast}(A)-A_{\infty}\|_{L^{2}_{1}(X)}\geq\tilde{\de},$$
unless $A$ is anti-self-dual  with respect to $g$.
\end{cor}
\section{K\"{a}hler surfaces}
In this section, we now denote by $X$ a compact K\"{a}hler surface with a K\"{a}hler form $\w$ and $E$ a principal $G$-bundle over $P$ with structure group $G$. We also set $d_{A}=\pa_{A}+\bar{\pa}_{A}$, $d^{\ast}_{A}=\pa^{\ast}_{A}+\bar{\pa}^{\ast}_{A}$ and $\phi=\sqrt{-1}(\theta-\theta^{\ast})$, where $\theta\in\Om^{1,0}(X,adE)$. Thus, Tanaka observed that Kapustin$\en$Witten equations on a closed K\"{a}hler surface are the same as Hitchin$\en$Simpson's equations, see \cite[Proposition 3.1]{TY2} .
\begin{prop}
Let $X$ be a closed K\"{a}hler surface, the Kapustin$\en$Witten equations have the following form that asks $(A,\theta)\in\mathcal{A}_{E}\times\Om^{1,0}(X,adE)$ to satisfies
\begin{equation*}
\begin{split}
&\bar{\pa}_{A}\theta=0,\ \theta\wedge\theta=0,\\
&F^{0,2}_{A}=0,\ \La_{\w}\big{(}F_{A}^{1,1}+[\theta\wedge\theta^{\ast}]\big{)}=0.\\
\end{split}
\end{equation*}
\end{prop}
Hence the bundle $E$ on $X$ is holomorphic and $\theta$ is a holomorphic section of $End(E)\otimes\Om^{1,0}(X)$, i.e., the bundle $(E,\theta)$ is a Higgs bundle. Following the Kobayashi$\en$Hitchin corresponding for Higgs bundle, see \cite{Hitchin,Simpson1988}, we know that the Higgs bundle $(E,\theta)$ is  stable since the pair $(A,\phi)$ satisfies the Hitchin$\en$Simpson equations.
\subsection{Irreducible connections}
In this section, we first recall a definition of irreducible connection on a principal $G$-bundle. Given a connection $A$ on a principal $G$-bundle $E$ over $X$. We can define the stabilizer $\Ga_{A}$ of $A$ in the gauge group $\mathcal{G}_{E}$ by
$$\Ga_{A}:=\{u\in\mathcal{G}_{E}|u^{\ast}(A)=A\},$$
one also can see \cite{DK,FU}. A connection $A$ called reducible if the connection $A$ whose stabilizer $\Ga_{A}$ is larger than the centre $C(G)$ of $G$. Otherwise, the connections are irreducible, they satisfy $\Ga_{A}\cong C(G)$. For the cases $G=SU(2)$ or $SO(3)$, it's easy to see that a connection $A$ is irreducible when it admits no nontrivial covariantly constant Lie algebra-value $0$-form, i.e.,
$$\ker d_{A}|_{\Om^{0}(X,adE)}=0.$$
We can defined the least eigenvalue $\la(A)$ of $d^{\ast}_{A}d_{A}$ as follow.
\begin{defi}\label{D2}
For $A\in\mathcal{A}_{E}$, define
\begin{equation}\label{E3}
\la(A):=\inf_{v\in\Om^{0}(X,adE)\backslash\{0\}}\frac{\|d_{A}v\|^{2}}{\|v\|^{2}}.
\end{equation}
is the lowest eigenvalue of $d^{\ast}_{A}d_{A}$.
\end{defi}
By the similar method of the proof of \cite[Proposition A.3]{PF2}  or \cite[Proposition 35.14]{PF1} , we would also show that the least eigenvalue $\la(A)$ of $d^{\ast}_{A}d_{A}$  with respect to connection $A$ is $L^{p}_{loc}$-continuity $2\leq p<4$.
\begin{prop}
Let $X$ be a closed, connected, oriented, smooth four-manifold with Riemannian metric, $g$. Let $\Sigma=\{x_{1},x_{2},\ldots,x_{L}\}\subset X$ ($L\in\N^{+}$) and $\rho=\min_{i\neq j}dist_{g}(x_{i},x_{j})$,\ let $U\subset X$ be the open subset give by
$$U:=X\backslash\bigcup_{l=1}^{L}\bar{B}_{\rho/2}(x_{l}).$$
Let $G$ be a compact Lie group, $A_{0}$, $A$ are  $C^{\infty}$ connections  on the principal $G$-bundles $E_{0}$ and $E$ over $X$ and $p\in[2,4)$.  There is an isomorphism of principal $G$-bundles, $u:E\upharpoonright X\backslash\Sigma\cong E_{0}\upharpoonright X\backslash\Sigma$, and identify $E\upharpoonright X\backslash\Sigma$ with $E_{0}\upharpoonright X\backslash\Sigma$ using this isomorphism. Then $\la(A)$ satisfies upper bound
\begin{equation}\nonumber
\sqrt{\la(A)}\geq\sqrt{\la(A_{0})}-c\sqrt{L}\rho^{1/6}(\la(A)+1)-cL\rho(\sqrt{\la(A)}+1)-c_{p}\|A-A_{p}\|_{L^{p}(U)}(\la(A)+1),
\end{equation}
and the lower bound,
\begin{equation}\nonumber
\sqrt{\la(A)}\leq\sqrt{\la(A_{0})}+c\sqrt{L}\rho^{1/6}(\la(A_{0})+1)+cL\rho(\sqrt{\la(A)}+1)+c_{p}\|A-A_{p}\|_{L^{p}(U)}(\la(A_{0})+1),
\end{equation}
where $c$ is a positive constant depends on $g,p$.
\end{prop}
We now have the useful 
\begin{cor}\label{C1}
Let $G$ be a compact Lie group and P a principal $G$-bundle over a closed, smooth, oriented, four-dimensional Riemannian manifold $X$ with a Riemannian metric $g$. If $\{A_{i}\}_{i\in\mathbb{N}}$ is a sequence $C^{\infty}$ connection on $P$ and the sequence $\{F^{+}_{A_{i}}\}_{i\in\N}$ converges to zero in $L^{2}$-topology, then
$$\lim_{i\rightarrow\infty}\la(A_{i})=\la(A_{\infty}).$$
where $\la(A)$ is as in Definition \ref{D2}.
\end{cor}
\begin{proof}
The proof is similar to the least eigenvalue of operator $d_{A}^{+}d_{A}^{+,\ast}|_{\Om^{2,+}(X,adE)}$ with respect to connection $A$ (see \cite[Corollary 35.16]{PF1} ). 
\end{proof}
For a compact K\"{a}hler surface $X$ we have a moduli space of ASD connections $M(E,g)$. The compacitification $\bar{M}(E,g)$ of $M(E,g)$ contained in the disjoint union
\begin{equation}
\bar{M}(E,g)\subset\cup(M(E_{l},g)\times Sym^{l}(X)),
\end{equation}
Following \cite[Theorem 4.4.3]{DK} , the space $\bar{M}(E,g)$ is compact.  
\begin{prop}\label{P2}
Let $X$ be a compact, K\"{a}hler surface with smooth K\"{a}hler metric $g$, $E$ be a principal $SU(2)$-bundle with $c_{2}(E)=c$ positive. Suppose that $g$ is a $c$-$general$ K\"{a}hler metric, then there are positive constants $\varepsilon=\varepsilon(E,g)$ and $\la_{0}=\la_{0}(E,g)$ with following significance.  If $A$ is a connection on $E$ such that
$$\|F^{+}_{A}\|_{L^{2}(X)}\leq\varepsilon,$$
and $\la(A)$ is as in (\ref{E3}), then 
$$\la(A)\geq\la_{0}.$$
\end{prop}
\begin{proof}
If $g$ is $c$-$generic$ metric, $\la(A)>0$ for $[A]\in M(E_{l},g)$ and $E_{l}$ is a principal $SU(2)$-bundle over $X$ appearing in the Uhlenbeck compactification. Since the function $\la(A)$, $A\in\bar{M}(E,g)$ is continuous by Corollary \ref{C1} and $\bar{M}(E,g)$ is compact, there is a uniform positive constant $\la$ such that $\la(A)\geq\la$ for $[A]\in\bar{M}(E,g)$. Suppose that the constant $\varepsilon$ does not exist. We may then choose a minimizing sequence $\{A_{i}\}_{i\in\N}$ of connections on $E$ such that $\|F_{A_{i}}^{+}\|_{L^{2}(X)}\rightarrow 0$ and $\la(A_{i})\rightarrow0$ as $i\rightarrow\infty$. According to Corollary \ref{C1}, $A_{i}$ converges to an ASD connection, $A_{\infty}$, on $L^{2}_{1,loc}(X)$. Then $\lim_{i\rightarrow\infty}\la(A_{i})=\la(A_{\infty})>0$, contradicting our initial assumption regarding the sequence $\{A_{i}\}_{i\in\ N}$.
\end{proof}
Fix an algebraic surface $S$ and an ample line bundle $L$ on $S$. For every integer $c$ we have defined the moduli space $\mathcal{M}_{c}(S,L)$ of $L$-stable rank two holomorphic vector bundles $V$ on $S$ and such that $c_{1}(V)=0$, $c_{2}(V)=c$. Taubes \cite{Taubes1984} has shown that, if $X$ is an arbitrary closed $4$-manifold and $g$ is any Riemannian metric on $X$, then the moduli space of all irreducible $g$-ASD connections on $P$ with $c_{2}(P)=c$ moduli gauge equivalence which denote by $\mathcal{M}(P_{c},g)$ is nonempty if $c$ is sufficiently large. In the case of an algebraic surface $S$ and a Hodge metric corresponding to the ample line bundle $L$, similar existence results are due to Maruyama and Gieseker \cite{Gieseker}. Friedman-Morgan (\cite[Chapter IV, Theorem 4.7]{Friedman-Morgan}) given a very general result along these lines is the following:
\begin{thm}
With $S$ and $L$ as above, for all $c\geq 2p_{g}(S)+2$, the modulis space $\mathcal{M}_{c}(S,L)$ is nonempty.
\end{thm} 
Next let us determine when a Hodge metric with K\"{a}hler form $\w$ admits reducible ASD connections. Corresponding to such a connection is an associated ASD harmonic $1$-form $\a$, well-defined up to $\pm1$, representing an integral cohomology class, which by the description of $\Om^{2,-}(X)$ is of type $(1,1)$ and orthogonal to $\w$. Thus Friedman-Morgan proved for an integer $c>0$, there exist Hodge metric $g$  over  $S$ is $c$-$generic$ in the sense of Definition \ref{D1} (see \cite[ Chapter IV, Proposition 4.8]{Friedman-Morgan}). For the convenience of readers, we give a detailed proof.
\begin{prop}\label{P4.7}
Fix $c>0$. Then there is an open dense subset $\mathcal{D}$ of the cone of ample divisors on $S$ such that if $g$ is a Hodge metric whose K\"{a}hler form lies in $\mathcal{D}$, $g$ is a $c$-$generic$ metric in the sense of Definition \ref{D1}.
\end{prop}
\begin{proof}
By standard argument, the set of $\w$ in the ample cone which are orthogonal to an integral class $\a$ with $0<-\a^{2}\leq c$ is the intersection of the ample cone with a collection of hyperplanes in $H^{2}(S;\mathbb{R})$ which is locally finite on the ample cone. This result is immediate from this. 
\end{proof}
\subsection{Approximate ASD connections}
In this section, we will give a general criteria under which an approximate ASD connection $A\in\mathcal{A}_{E}^{1,1}$ can be deformation into an other approximate ASD connection $A_{\infty}$ which obeys the equation
\begin{equation}\label{E0}
\La_{\w}F_{A_{\infty}}=0.
\end{equation} 
Let $A$ be a connection on a principal $G$-bundle over $X$. The  equation (\ref{E0}) for a second  connection $A_{\infty}:=A+a$, where $a\in\Om^{1}(X,adE)$ is a bundle valued $1$-form, can be written as:
\begin{equation}\label{E1}
\La_{\w}(d_{A}a+a\wedge a)=-\La_{\w}F_{A}.
\end{equation}
We seek a solution of the equation (\ref{E1}) in the form 
$$a=d_{A}^{\ast}(s\otimes\w)=\sqrt{-1}(\pa_{A}s-\bar{\pa}_{A}s)$$
where $s\in\Om^{0}(X,adE)$ is a section on $ad(E)$. Then equation (\ref{E1}) becomes to a second order equation:
\begin{equation}\label{E2}
-d_{A}^{\ast}d_{A}s+\La_{\w}(d_{A}s\wedge d_{A}s)=-\La_{\w}F_{A}.
\end{equation}
For convenience, we define a map
$$B(u,v):=\frac{1}{2}\La_{\w}[d_{A}u\wedge d_{A}v].$$
It's easy to check, we have the pointwise bound:
$$|B(u,v)|\leq C|\na_{A}u||\na_{A}v|,$$
where $C$ is a uniform positive constant. We want to prove that if $\La_{\w}F_{A}$ is small in an appropriate sense, then there is a small solution $s$ to equation (\ref{E2}).
\begin{thm}\label{T6}
Let $X$ be a compact K\"{a}hler surface with a K\"{a}hler metric $g$, $E$ a principal $SU(2)$-bundle over $X$. There are positive constants  $\varepsilon=\varepsilon(E,g)\in(0,1)$ and $\la=\la(E,g)\in(0,\infty)$ with following significance. If $A\in\mathcal{A}_{E}^{1,1}$ obeys 
\begin{equation}\nonumber
\begin{split}
&\|\La_{\w}F_{A}\|_{L^{2}(X)}\leq\varepsilon,\\
&\la(A)\geq\la,\\
\end{split}
\end{equation}
then there is a section $s\in\Om^{0}(X,adE)$ such that the connection $A_{\infty}:=A+\sqrt{-1}(\pa_{A}s-\bar{\pa}_{A}s)$ satisfies\\
(1) $\La_{\w}F_{A_{\infty}}=0$\\
(2) $\|s\|_{L^{2}_{2}(X)}\leq C\|\La_{\w}F_{A}\|_{L^{2}(X)}$,\\
(3) $\|F^{0,2}_{A_{\infty}}\|_{L^{2}(X)}\leq C\|\La_{\w}F_{A}\|^{2}_{L^{2}(X)}$.\\
where $C=C(\la,g)\in[1,\infty)$ is a positive constant.
\end{thm}
Now, we begin to prove Theorem \ref{T6}. The method of proof Theorem \ref{T6} is base on Taubes' ideas \cite{T4}.  At first, suppose that $s$ and $f$ are sections of $adE$ with
\begin{equation}\label{E8}
d^{\ast}_{A}d_{A}s=f,\  i.e.,\ \na_{A}^{\ast}\na_{A}s=f. 
\end{equation}
The first observation is
\begin{lem}\label{L2}
If $\la(A)\geq\la>0$, then there exists a unique $C^{\infty}$ solution to equation (\ref{E8}). Furthermore, we have
\begin{equation*}
\begin{split}
&\|s\|_{L^{2}_{2}(X)}\leq c\|f\|_{L^{2}(X)},\\
&\|B(s,s)\|_{L^{2}(X)}\leq c\|f\|^{2}_{L^{2}(X)},\\
\end{split}
\end{equation*}
where $c=c(\la,g)$ is a positive constant.
\end{lem} 
\begin{proof}
Since $\na_{A}^{\ast}\na_{A}$ is an elliptic operator of order $2$, then for each $k\geq 0$, there is a positive constant $C_{k}$ so that for all section $v$ of $adE$, see \cite{DK} (A8),
\begin{equation*}
\begin{split}
\|v\|_{L^{2}_{k+2}(X)}&\leq C_{k}(\|\na^{\ast}_{A}\na_{A}v\|_{L^{2}_{k}(X)}+\|v\|_{L^{2}(X)})\\ 
&\leq C_{k}(\|\na^{\ast}_{A}\na_{A}v\|_{L^{2}_{k}(X)}+\la^{-1}\|\na_{A}^{\ast}\na_{A}v\|_{L^{2}(X)}).\\
\end{split}
\end{equation*}
We take $v=s$ and $k=0$, then
$$\|s\|_{L^{2}_{2}(X)}\leq c\|\na_{A}^{\ast}\na_{A}s\|_{L^{2}(X)}\leq c\|f\|_{L^{2}(X)},$$
where $c=c(\la,g)$ is a positive constant. By the Sobolev inequality in four dimension,
$$\|B(s,s)\|_{L^{2}(X)}\leq C\|\na_{A}s\|^{2}_{L^{4}(X)}\leq C\|\na_{A}s\|^{2}_{L^{2}_{1}(X)}\leq c\|s\|_{L^{2}_{2}(X)},$$
where $C$ is a positive constant. Hence we complete the proof of this lemma.
\end{proof}
\begin{lem}\label{L3}
If $d_{A}^{\ast}d_{A}s_{1}=f_{1}$, $d_{A}^{\ast}d_{A}s_{2}=f_{2}$, then
\begin{equation*}
\|B(s_{1},s_{2})\|_{L^{2}(X)}\leq c\|f_{1}\|_{L^{2}(X)}\|f_{2}\|_{L^{2}(X)}.
\end{equation*}
\end{lem}
We can prove the existence of a solution of (\ref{E2}) by the contraction mapping principle. We write $s=(d_{A}^{\ast}d_{A})^{-1}f$ and (\ref{E2}) becomes an equation for $f$ of the from 
\begin{equation}\label{E10}
f-S(f,f)=\La_{\w}F_{A},
\end{equation}
where  $S(f,g):=B((d_{A}^{\ast}d_{A})^{-1}f,(d_{A}^{\ast}d_{A})^{-1}g)$. By Lemma \ref{L3},
\begin{equation*}
\begin{split}
\|S(f_{1},f_{1})-S(f_{2},f_{2})\|_{L^{2}(X)}&=\|S(f_{1}+f_{2},f_{1}-f_{2})\|_{L^{2}(X)}\\
&\leq c\|f_{1}+f_{2}\|_{L^{2}(X)}\|f_{1}-f_{2}\|_{L^{2}(X)}.\\
\end{split}
\end{equation*}
We denote $g_{k}=f_{k}-f_{k-1}$ and $g_{1}=f_{1}$, then
$$g_{1}=\La_{\w}F_{A},\  g_{2}=S(g_{1},g_{1})$$
and $$g_{k}=S(\sum_{i=1}^{k-1}g_{i},\sum_{i=1}^{k-1}g_{i})-S(\sum_{i=1}^{k-2}g_{i},\sum_{i=1}^{k-2}g_{i}),\ \forall\ k\geq 3.$$
It is easy to show that, under the assumption of $\La_{\w}F_{A}$, the sequence $f_{k}$ defined by
$$f_{k}=S(f_{k-1},f_{k-1})+\La_{\w}F_{A},$$
starting with $f_{1}=\La_{\w}F_{A}$, is Cauchy with respect to $L^{2}$, and so converges to a limit $f$ in the completion of $\Ga(adE)$ under $L^{2}$.
\begin{prop}\label{P3}
There are positive constant $\varepsilon\in(0,1)$ and $C\in(1,\infty)$ with following significance. If the connection $A$ satisfies $$\|\La_{\w}F_{A}\|_{L^{2}(X)}\leq\varepsilon,$$
then each $g_{k}$ exists and is $C^{\infty}$. Further for each $k\geq1$, we have
\begin{equation}\label{E9}
\|g_{k}\|_{L^{2}(X)}\leq C^{k-1}\|\La_{\w}F_{A}\|^{k}_{L^{2}(X)}.
\end{equation}
\end{prop}
\begin{proof}
The proof is by induction on the integer $k$. The induction begins with $k=1$, one can see $g_{1}=\La_{\w}F_{A}$. The induction proof if completed by demonstrating that if (\ref{E9}) is satisfied for $j<k$, then it also satisfied for $j=k$. Indeed, since 
\begin{equation}\nonumber
\begin{split}
&\quad\|S(\sum_{i=1}^{k-1}g_{i},\sum_{i=1}^{k-1}g_{i})-S(\sum_{i=1}^{k-2}g_{i},\sum_{i=1}^{k-2}g_{i})\|_{L^{2}(X)}\\
&\leq c\|\sum_{i=1}^{k-1}g_{i}+\sum_{i=1}^{k-2}g_{i}\|_{L^{2}(X)}\|g_{k-1}\|_{L^{2}(X)},\\
&\leq 2c\sum\|g_{i}\|_{L^{2}(X)}\|g_{k-1}\|_{L^{2}(X)},\\
&\leq \frac{2c}{1-C\|\La_{\w}F_{A}\|_{L^{2}(X)}}C^{k-2}\|\La_{w}F_{A}\|^{k}_{L^{2}(X)}.\\
\end{split}
\end{equation}
Now, we provide the constants $\varepsilon$ sufficiently small and $C$ sufficiently large to ensures that
$$\|\La_{\w}F_{A}\|_{L^{2}(X)}\leq C^{-2}(C-2c),$$
i.e., $\frac{2c}{1-C\|\La_{w}F_{A}\|_{L^{2}(X)}}\leq C$,  hence  we complete the proof of this proposition.
\end{proof}
\begin{proof}[\textbf{Proof of Theorem \ref{T6}}.] The sequence $g_{k}$ is Cauchy in $L^{2}$, the limit $$f:=\lim_{i\rightarrow\infty}f_{i}$$ is a solution to (\ref{E10}). Using Lemma \ref{L2} and Proposition \ref{P3}, we have
$$\|s\|_{L^{2}_{2}(X)}\leq C\|f\|_{L^{2}(X)}\leq\frac{\|\La_{\w}F_{A}\|_{L^{2}(X)}}{1-C\|\La_{\w}F_{A}\|_{L^{2}(X)}},$$
we provide $\varepsilon$ and $C$ to ensure $C\varepsilon\leq\frac{1}{2}$, hence 
$$\|s\|_{L^{2}_{2}(X)}\leq 2\|\La_{\w}F_{A}\|_{L^{2}(X)}.$$
We denote $A_{\infty}:=A+\sqrt{-1}(\pa_{A}s-\bar{\pa}_{A}s)$, then 
\begin{equation*}
\begin{split}
\|F^{0,2}_{A_{\infty}}\|_{L^{2}(X)}&=\|-\sqrt{-1}\bar{\pa}_{A}\bar{\pa}_{A}s-\bar{\pa}_{A}s\wedge\bar{\pa}_{A}s\|_{L^{2}(X)}\\
&=\|\bar{\pa}_{A}s\wedge\bar{\pa}_{A}s\|_{L^{2}(X)}\leq 2\|\bar{\pa}_{A}s\|^{2}_{L^{4}(X)}\\
&\leq c\|\na_{A}s\|^{2}_{L^{4}(X)}\leq c\|\na_{A}s\|^{2}_{L^{2}_{1}(X)}\\
&\leq c\|\La_{\w}F_{A}\|^{2}_{L^{2}(X)}.\\
\end{split}
\end{equation*}
where $c$ is a positive constant. We complete the proof of Theorem \ref{T6}.
\end{proof}
\subsection{A topological property for the moduli space of stable Higgs bundles}
In this section, we use the inequality in Theorem \ref{T6} to prove that the Higgs filed has a positive lower bounded if the K\"{a}hler metric $g$ is $c$-$generic$. Following Theorem \ref{T2.6}, we have a vanishing result  of Higgs fields as follows.
\begin{cor}
Let $X$ be a closed K\"{a}hler surface with a smooth K\"{a}hler metric $g$,  $E$ be a principal $SU(2)$-bundle with $c_{2}(E)=c$ positive. Suppose that $g$ is a $c-generic$ metric in the sense of Definition \ref{D1}. If the Higgs pair $(A,\theta)$ is a smooth solution of the equations: $\La_{\w}F_{A}=0,\ [\theta\wedge\theta^{\ast}]=0$, then the Higgs field $\theta$ vanishes.
\end{cor}
\begin{proof}[\textbf{Proof Theorem \ref{T1}}.] Let  $(A,\theta)\in\mathcal{A}^{1,1}_{E}\times\Om^{1,0}(X,\mathfrak{g}_{E})$ be a smooth solution of Hitchin$\en$Simpson equations. We denote $\phi=\sqrt{-1}(\theta-\theta^{\ast})\in\Om^{1}(X,adE)$. Suppose that the constant $C$ does not exist. We may provide a positive constant $\varepsilon_{0}$ such that 
$$0<\|\La_{\w}F_{A}\|_{L^{2}(X)}\leq c\|\theta\|^{2}_{L^{2}(X)}\leq c\varepsilon_{0}<\varepsilon$$
where $c=c(g)$ is a positive constant and $\varepsilon$ is a constant as Proposition \ref{P2}. Thus $\la(A)>0$. Following Theorem \ref{T6}, there exist a connection $A_{\infty}$ which satisfies $\La_{w}F_{A_{\infty}}=0$ and 
$$\|A-A_{\infty}\|_{L^{2}(X)}\leq c\|\La_{\w}F_{A}\|_{L^{2}(X)}$$
for some positive constant $c=c(g,P)$. We have two integrable inequalities:
$$\|\na_{A}\phi\|^{2}_{L^{2}(X)}+\langle Ric\circ\phi,\phi\rangle_{L^{2}(X)}+\|\La_{\w}F_{A}\|^{2}_{L^{2}(X)}=0,$$
$$\|\na_{A_{\infty}}\phi\|^{2}_{L^{2}(X)}+\langle Ric\circ\phi,\phi\rangle_{L^{2}(X)}+\langle F^{0,2}_{A_{\infty}}+F^{2,0}_{A_{\infty}},[\phi,\phi]\rangle_{L^{2}(X)}\geq0.$$
To obtain the second identity, we use the Weitzenb\"{o}ck formula  \cite{FU} Equation (6.25):
$$2d_{A_{\infty}}^{-,\ast}d_{A_{\infty}}^{-}+d_{A_{\infty}}d_{A_{\infty}}^{\ast}=\na^{\ast}_{A_{\infty}}\na_{A_{\infty}}+Ric(\cdot)+\ast[F_{A_{\infty}}^{+},\cdot]=0\ on\ \Om^{1}(X,\mathfrak{g}_{P}).$$
We also observe that $\langle F^{0,2}_{A_{\infty}}+F^{2,0}_{A_{\infty}},[\phi,\phi]\rangle_{L^{2}(X)}=0$ since $[\phi,\phi]\in\Om^{1,1}(X,adE)$.

Combining the preceding inequalities gives
\begin{equation}\nonumber
\begin{split}
0&\leq\|\na_{A_{\infty}}\phi\|^{2}_{L^{2}(X)}+\langle Ric\circ\phi,\phi\rangle_{L^{2}(X)}\\
&\leq\|\na_{A}\phi\|^{2}_{L^{2}(X)}+\langle Ric\circ\phi,\phi\rangle_{L^{2}(X)}+\|\na_{A}\phi-\na_{A_{\infty}}\phi\|_{L^{2}(X)}^{2}\\
&\leq(c\|\phi\|^{2}_{L^{2}(X)}-1)\|\La_{\w}F_{A}\|^{2}_{L^{2}(X)}.\\
\end{split}
\end{equation}
for a positive constant $c=c(g)$. If we provide $C=1/(2c)$ such that $\|\phi\|^{2}_{L^{2}(X)}\leq C$, then $\La_{\w}F_{A}\equiv0$. Its contradiction to our initial assumption regarding the $\La_{\w}F_{A}$. The preceding argument shows that the desired constant $C$ exists.

We also suppose $X$ is simply-connected, if the constant $\tilde{C}$ does not exist. We may choose a sequence $\{(A_{i},\theta_{i})\}_{i\in\N}$ of Hitchin$\en$Simpson equations such that
$\|\La_{\w}F_{A_{i}}\|_{L^{2}(X)}\rightarrow0$. We set $\phi_{i}:=\sqrt{-1}(\theta_{i}-\theta^{\ast}_{i})$.Following the idea of Theorem \ref{T1.2}, the Proposition \ref{P3.11} implies that there exist a subsequence $\Xi\in\N$ and a sequence transformation $\{g_{i}\}_{i\in\Xi}$ such that 
$(g^{\ast}(A_{i}),g^{\ast}(\phi_{i}))\rightarrow (A_{\infty},0)$ in $C^{\infty}$ over $X-\Sigma$, $\Sigma:=\{x_{1},\cdots,x_{L}\}$. We  also  have $\lim_{i\rightarrow\infty}\int_{X}|\phi_{i}|^{2}=0$. This contradict $\|\phi_{i}\|_{L^{2}(X)}$ has a uniform positive lower bounded. The preceding argument shows that the desired constant $\tilde{C}$ exists.
\end{proof}
\begin{cor}\label{C4.12}
Let $X$ be a compact, simply-connected, K\"{a}hler surface with a smooth K\"{a}hler metric $g$, $E$ be a principal $SU(2)$-bundle over $X$. Suppose that $(E,\theta)$ is a stable Higgs bundle with $c_{2}(E)=c$, then there is a positive integer $c$ with following significance. If $g$ is a $c$-$generic$ metric in the sense of Definition (\ref{D1}), then the moduli space of stable Higgs bundle is non-connected. 	
\end{cor}
\begin{proof}
Taubes \cite{Taubes1984} has shown that, if $X$ is an arbitrary closed $4$-manifold and $g$ is any Riemannian metric on $X$, $P$ be a $SU(2)$-bundle with $c_{2}(E)=c$ sufficiently large, then there exist an irreducible ASD connection on $P$. The moduli space  $M_{ASD}(E,g)$ is non-empty under our conditions. Since the map $(A,\theta)\mapsto\|\theta\|_{L^{2}}$ is continuous, then the moduli space $M_{HS}$ is not connected.
\end{proof}
Following Proposition \ref{P4.7} and Corollary \ref{C4.12}, we have
\begin{cor}
Let $S$ be a compact, simply-connected, algebraic surface with a Hodge metric $g$, $(E,\theta)$ be a Higgs bundle over $X$. There is an open dense subset $\mathcal{D}$ of the cone of ample divisor on $S$ and a positive integer $c$ with following significance. If the K\"{a}hler form of $g$ lies in $\mathcal{D}$ and $(E,\theta)$ is stable with $c_{2}(E)=c$, then the moduli space of stable Higgs bundle is non-connected.
\end{cor}
\subsection{Vafa-Witten equaions}
We next consider analogue of the Hitchin$\en$Simpson equations, called Vafa$\en$Witten equations \cite{VW}. Let $X$ be a closed, oriented, smooth, four-manifold, $E$ be a principal $G$-bundle with $G$ being a compact Lie group over $X$. We denote by $A$ a connection on $E$, $B$ a section of the associated bundle $\Om^{2,+}\otimes adE$ and $C$ a section of $adE$. A triple $(A,B,C)$ is a solution of Vafa$\en$Witten equations if $(A,B,C)$ satisfies 
\begin{equation}\label{E6}
F_{A}^{+}+[B.B]+[B,\Ga]=0,\ d_{A}^{+,\ast}B+d_{A}\Ga=0,
\end{equation}  
where $[B.B]\in\Ga(X,\Om^{2,+}\otimes adE)$ is defined through the Lie brackets of $adE$ and $\Om^{2,+}$, see \cite[Appendix A]{BM} . We take $X$ to be a compact K\"{a}hler surface with K\"{a}hler form $\w$. In a local orthonomal coordinates, we can write $B=\b-\b^{\ast}+B_{0}\w$, where $\b\in\Om^{2,0}(X,adE)$, $B_{0}\in\Om^{0}(X,adE)$ and $\gamma=C-\sqrt{-1}B_{0}$, one can see \cite[Chapter 7]{BM}  for details. We have
\begin{thm}(\cite[Theorem 7.1.2]{BM} )
On a compact K\"{a}hler surface, the Vafa$\en$Witten equations (\ref{E6}) have the following form that asks $(A,\b,\gamma)\in\mathcal{A}_{E}^{1,1}\times\Om^{2,0}(X,adE)\times\Om^{0}(X,adE)$ to satisfies
\begin{equation}\label{E7}
\begin{split}
&\sqrt{-1}\w\wedge F_{A}+\frac{1}{2}[\b\wedge\b^{\ast}]=0,\\
&\bar{\pa}_{A}\b=d_{A}\gamma=0,\\
&[\gamma,\gamma^{\ast}]=[\gamma,\b+\b^{\ast}]=0.\\
\end{split}
\end{equation}
\end{thm}
Mares also discussed a relation between the existence of a solution to the equations and a stability of vector bundles, see \cite{BM,Tanaka2015}.
\begin{defi}
For any $\b\in\Om^{2,0}(X,End E)\cong\Om^{0}(Hom(E,E\otimes K))$, we say that a subbundle $E'\subset E$ is $\b$-invariant, if
$$\b(E')\subset E'\otimes K.$$
A holomorpic bundle $E$ is $\b$-stable (semi-stable) if all $\b$-invariant holomorphic subbundle $E'\subset E$ satisfies
$$\mu(E')=\frac{deg(E')}{rankE'}<(\leq)\mu(E)=\frac{deg(E)}{rankE}.$$
\end{defi}
We also recall a bound on $\|\b\|_{L^{\infty}}$ in terms of $\|\b\|_{L^{2}}$. One can see \cite{BM} for detail.
\begin{thm}
Let $X$ be a compact, K\"{a}hler surface with a smooth K\"{a}hler metric $g$,\ $E$ a  principal $G$-bundle over $X$ with $G$ being a compact Lie group. If the pair $(A,\b,\gamma)$ is a smooth solution of Equations (\ref{E7}), there is a positive constant $C=C(g)$ such that
\begin{equation}\nonumber
\|\b\|_{L^{\infty}(X)}\leq C\|\b\|_{L^{2}(X)}.
\end{equation}
\end{thm}
We extent the idea of Hitchin$\en$Simpson equations to the case of the Vafa$\en$Witten equations on the compact K\"{a}hler surface with a $c$-$generic$ K\"{a}hler metric.
\begin{thm}
Let $X$ be a compact, K\"{a}hler surface with a smooth K\"{a}hler metric $g$, $E$ be a principal $SU(2)$-bundle with $c_{2}(E)=c$ positive. Suppose that $g$ is a $c$-$generic$ metric in the sense of Definition (\ref{D1}), then there is a positive constant $C=C(g,P)$ with following significance. If the triple $(A,\b,\gamma)\in\mathcal{A}_{E}^{1,1}\times\Om^{2,0}(X,adE)\times\Om^{0}(X,adE)$ is satisfies  equations (\ref{E7}), then one of following must hold:\\
(1) $\La_{\w}F_{A}=0$, or\\
(2) the extra field $\b$ satisfies 
$$\|\b\|_{L^{2}(X)}\geq C.$$
\end{thm}
\begin{proof}
For $(A,\b,\gamma)$ is a solution of equations (\ref{E7}), we denote $\b_{0}:=\b-\b^{\ast}\in\Om^{2,+}(X,adE)$, hence $d_{A}^{+,\ast}\b_{0}=\ast(\bar{\pa}_{A}\b-\pa_{A}\bar{\b})=0$.
Suppose that the constant $C$ does not exist. We may provide a positive constant $\varepsilon_{0}$ such that 
$$0<\|\La_{\w}F_{A}\|_{L^{2}(X)}\leq c\|\b\|^{2}_{L^{2}(X)}\leq c\varepsilon_{0}<\varepsilon$$
where $c=c(g)$ is a positive constant and $\varepsilon$ is a constant as Proposition \ref{P2}. Thus $\la(A)>0$. Following Theorem \ref{T6}, there exist a connection $A_{\infty}$ such that
$$\|A-A_{\infty}\|_{L^{2}(X)}\leq c\|\La_{\w}F_{A}\|_{L^{2}(X)}$$
for a positive constant $c=c(g,P)$. By the Weizenb\"{o}ck formula (6.26) in \cite{FU}, we have two integrable identities:
$$\|\na_{A}\b_{0}\|^{2}_{L^{2}(X)}+\langle (\frac{1}{3}s-2\w^{+})\b_{0},\b_{0}\rangle_{L^{2}(X)}+\|\La_{\w}F_{A}\|^{2}_{L^{2}(X)}=0,
$$
$$\|\na_{A_{\infty}}\b_{0}\|^{2}_{L^{2}(X)}+\langle (\frac{1}{3}s-2\w^{+})\b_{0},\b_{0}\rangle_{L^{2}(X)}+
\langle F_{A_{\infty}}^{2,0}+F_{A_{\infty}}^{0,2},[\b_{0}.\b_{0}]\rangle\geq0.$$
We also observe that 
\begin{equation*}
\begin{split}
\langle F^{0,2}_{A}+F_{A}^{2,0},[\b_{0}.\b_{0}]\rangle_{L^{2}(X)}&\leq c\|F_{A}^{0,2}\|_{L^{2}(X)}\|\b_{0}\|^{2}_{L^{4}(X)}\\
&\leq c\|\La_{\w}F_{A}\|^{2}_{L^{2}(X)}\|\b_{0}\|^{2}_{L^{2}(X)},\\
\end{split}
\end{equation*}
for a positive constant $c=c(g)$. Combing the preceding inequalities gives
\begin{equation}\nonumber
\begin{split}
&0\leq\|\na_{A_{\infty}}\b_{0}\|^{2}_{L^{2}(X)}+\langle (\frac{1}{3}s-2\w^{+})\b_{0},\b_{0}\rangle_{L^{2}(X)}+
\langle F_{A_{\infty}}^{2,0}+F_{A_{\infty}}^{0,2},[\b_{0}.\b_{0}]\rangle\\
&\leq\|\na_{A}\b_{0}\|^{2}_{L^{2}(X)}+\langle (\frac{1}{3}s-2\w^{+})\b_{0},\b_{0}\rangle_{L^{2}(X)}\\
&+\|\na_{A}\b_{0}-\na_{A_{\infty}}\b_{0}\|_{L^{2}(X)}^{2}+c\|\La_{\w}F_{A}\|^{2}_{L^{2}(X)}\|\b_{0}\|^{2}_{L^{2}(X)}\\
&\leq\|[A-A_{\infty},\b_{0}]\|^{2}_{L^{2}(X)}+c\|F^{0,2}_{A}\|_{L^{2}(X)}\|\b_{0}\|^{2}_{L^{\infty}(X)}-\|\La_{\w}F_{A}\|^{2}_{L^{2}(X)}\\
&\leq(c\|\b_{0}\|^{2}_{L^{2}(X)}-1)\|\La_{\w}F_{A}\|^{2}_{L^{2}(X)}.\\
\end{split}
\end{equation}
for a positive constant $c=c(g)$. We provide  $\|\b_{0}\|^{2}_{L^{2}(X)}\leq 1/(2c)$, thus $\La_{\w}F_{A}=0$. This contradicting our initial assumption regarding the $\La_{\w}F_{A}$. The preceding argument shows that the desired constant $C$ exists.
\end{proof}
We denote by  
$$M_{VW}:=\{(A,\b,\gamma)\in\mathcal{A}_{E}^{1,1}\times\Om^{2,0}(adE)\times\Om^{0}(adE)\mid(A,\b,\gamma)\ satisfies\ (\ref{E7})\}/\mathcal{G}_{E}$$
the moduli space of solutions of Vafa$\en$Witten equations. We apply a vanishing theorem of extra fields of Vafa$\en$Witten equations due to Mare \cite[Theorem 4.2.]{BM} 1 to prove that
\begin{thm}
Let $X$ be a simply-connected, K\"{a}hler surface with a smooth K\"{a}hler metric $g$, $E$ be a principal $SU(2)$ or $SO(3)$-bundle. Let the triple $(A,\b,\gamma)$ be a solution of the $decoupled$ Vafa$\en$Witten equations:
\begin{equation*}
\begin{split}
&\La_{\w}F_{A}=0,\\
&\bar{\pa}_{A}\b=[\b\wedge\b^{\ast}]=0,\\
&d_{A}\gamma=[\gamma,\gamma^{\ast}]=[\gamma,\b+\b^{\ast}]=0.\\
\end{split}
\end{equation*}
If $A$ is an irreducible connection, then $\b=\gamma=0$.
\end{thm}
Following the idea of Hitchin$\en$Simpson equations, we also have
\begin{cor}
Let $X$ be a compact, simply-connected, K\"{a}hler surface with a smooth K\"{a}hler metric $g$, $E$ be a principal $SU(2)$-bundle.  Suppose that $(E,\b)$ is $\b$-stable bundle with $c_{2}(E)=c$, then there is a positive integer $c$ with following significance. If $g$ is a  $c$-$generic$ metric in the sense of Definition (\ref{D1}), then the moduli space of $\b$-stable bundle is non-connected. 	
\end{cor}
\section*{Acknowledgment}
I would like to thank the anonymous referee for  careful reading of my manuscript and helpful comments. I would like to thank Professor Clifford Taubes for helpful comments regarding his article \cite{T1}. This work is supported by Nature Science Foundation of China No. 11801539 and Postdoctoral Science Foundation of China No. 2017M621998, No. 2018T110616.

\end{document}